\documentclass[journal,10pt]{IEEEtran}
\usepackage{stfloats}
\usepackage{amsmath}
\usepackage{amssymb,latexsym}
\usepackage{graphicx}
\usepackage{color,cite,times}
\usepackage{epstopdf}
\usepackage{algorithmic}
\usepackage[ruled,linesnumbered] {algorithm2e}
\usepackage{booktabs}
\usepackage{bbm} 
\usepackage{url}
\usepackage{fancyhdr}
\usepackage{verbatim}
\usepackage{multirow}
\usepackage{makecell}
\usepackage{graphicx}
\usepackage{indentfirst}
\usepackage{cases}
\newcommand{\RNum}[1]{\uppercase\expandafter{\romannumeral #1\relax}}
\usepackage{colortbl} 
\usepackage{extarrows}
\usepackage{threeparttable}
\usepackage{diagbox}
\usepackage{pifont}
\usepackage{cancel}
\usepackage{framed}
\usepackage{overpic}
\usepackage{subfigure}
\usepackage{amsthm}
\usepackage{enumitem}

\usepackage[framemethod=tikz]{mdframed}
\usepackage{lipsum}
\usepackage[nolist]{acronym}
\begin{acronym}

\acro{SAA}{sample average approximation}
\acro{6G}{sixth-generation}
\acro{AI}{artificial intelligence}

\acro{E2E}{end-to-end}
\acro{OFDMA}{orthogonal frequency-division multiple access}
\acro{FDD}{frequency-division duplexing}
\acro{UL}{uplink}
\acro{DL}{downlink}
\acro{SPIN}{\emph{speculative inference}}
\acro{Multi-SPIN}{\emph{multi-access SPIN}}
\acro{SLM}{\emph{small language model}}
\acro{LLM}{\emph{large language model}}
\acro{w.r.t.}{with respect to}

\acro{MSE}{mean-squared error}
\acro{MMSE}{minimum mean-squared error}
\acro{NMSE}{normalized MSE}
\acro{LMMSE}{linear MMSE}
\acro{SVD}{singular value decomposition}
\acro{KL}{Kullback-Leibler}
\acro{KKT}{Karush-Kuhn-Tucker}
\acro{AWGN}{additive white Gaussian noise}
\acro{i.i.d.}{independent and identically distributed}
\end{acronym}

\usepackage{tcolorbox}
\tcbset{colback=gray!15,colframe=gray!15,boxrule=0pt,arc=2pt}

\newcommand{\bsm}{\boldsymbol}

\newcommand{\mbf}{\mathbf}

\theoremstyle{definition}  

\newtheorem{lemm}{Lemma}
\newtheorem{remk}{Remark}

\newtheorem{prop}{Proposition}

\newtheorem{observ}{Observation}

\definecolor{gray}{RGB}{192,192,192}

\begin{document}
\title{Multi-Access Speculative Inference: \\Uplink or Downlink?}

\author{Chang Cai, 
\IEEEmembership{Member, IEEE,}
	Kaibin Huang,
    \IEEEmembership{Fellow, IEEE}
		\thanks{C. Cai and K. Huang are with the Department of Electrical and Computer Engineering, The University of Hong Kong, Hong Kong SAR (e-mail: changcai@hku.hk; huangkb@hku.hk).
        Corresponding author: K. Huang.
		}
	}
\maketitle

\begin{abstract}

    Multi-access speculative inference (Multi-SPIN) extends SPIN to multi-device edge networks to accelerate cooperative token generation.
    It allows on-device small language models (SLMs) to autoregressively draft multiple tokens for individual generation tasks, while an edge-server large language model (LLM) verifies them in parallel.
    The major communication overhead arises when a drafted token is rejected by the server, in which case sampling the correction token requires access to both the SLM-output draft distribution and the LLM-output target distribution over the full token vocabulary.
    Existing designs typically perform correction at the server by uploading the draft distribution, but transmitting a vocabulary-wide distribution creates a critical uplink (UL) bottleneck.
    Alternatively, the correction can be performed at the device by downloading the target distribution, leveraging the high transmission rates available on the downlink (DL).
    Motivated by this insight, we introduce communication-mode selection as a new design dimension for Multi-SPIN.
    Specifically, each device can adaptively switch between the UL and DL modes to balance the UL bottleneck against the shared DL resource constraint, thereby relieving the overall communication burden.
    We formulate a sum-token-goodput maximization problem that jointly accounts for mode selection, draft-length control, and power allocation.
    For mode selection, we reveal a simple structure of its optimal solution, which enables an efficient search over the number of devices assigned to the UL mode, with the optimal transmit powers determined accordingly.
    For draft-length control, we develop a greedy-search-based algorithm that adapts the device-specific draft lengths to heterogeneous computation and communication capabilities.
    Experimental results on both Qwen2.5 and DeepSeek-R1 model pairs across diverse tasks demonstrate that the proposed framework significantly improves token goodput by flexibly exploiting UL and DL resources.
    
\end{abstract}


\section{Introduction}

Deploying \acp{LLM} at the wireless network edge enables inference to be performed closer to end users, offering a promising approach to supporting low-latency and privacy-preserving \ac{AI} services~\cite{edgeai2022lulu, jiaweishao2026aiflow}.
On one hand, running large models directly on edge devices remains fundamentally challenging due to their stringent computation, memory, and energy constraints.
On the other hand, relying exclusively on an edge server to provide \ac{LLM} services can create a centralized inference bottleneck while leaving distributed computing resources at devices underutilized.
This bottleneck is further exacerbated by the autoregressive nature of \ac{LLM} decoding, where tokens are generated sequentially and thus a separate forward pass for each token is needed.
To address these limitations, \ac{SPIN}~\cite{SDgoogle2023icml, SDdeepmind2023} has recently emerged as a promising paradigm for collaborative edge \ac{LLM} inference.
The key idea of \ac{SPIN} is to deploy a lightweight \ac{SLM} at the device to generate multiple draft tokens locally, and then use the server-side \ac{LLM} to verify the drafted block in parallel~\cite{spin2025infocom, edgellm2025tmc}.
Because evaluating a sequence of drafted tokens concurrently requires only a single forward pass, \ac{SPIN} significantly reduces the computational overhead at the server without sacrificing \ac{LLM}-level inference quality. 
Furthermore, by offloading draft generation to devices, \ac{SPIN} effectively harnesses the computing resources distributed across the network edge.

Existing efforts on collaborative edge inference have primarily focused on split inference~\cite{split_inference2020xuchen, edge_inference2020jiaweishao}, where a large model is partitioned between the device and the edge server.
In this paradigm, the device executes the early layers of the model and transmits intermediate features to the server, which then completes the remaining computation~\cite{cai2024mcr2, wen2024edgeAI, cai2025e2e}.
By dynamically adjusting the split point, split inference can balance the computation workloads between the devices and the server as well as control the wireless transmission overhead~\cite{split2022jiayan}.
However, this layer-wise partitioning strategy faces inherent limitations when applied to autoregressive \ac{LLM} decoding.
Since the generation of each token depends on all preceding tokens, the device and the server must repeatedly execute their respective model partitions and exchange intermediate features at every single decoding step~\cite{zhangkai2025jstsp}.
This continuous back-and-forth results in excessive computation and communication delays, limiting the suitability of split inference for latency-sensitive edge \ac{LLM} services.
In contrast, \ac{SPIN} introduces a fundamental paradigm shift from sequential token generation to parallel token verification.
By deploying a low-complexity \ac{SLM} at the device to speculate on multiple future tokens, several inherently sequential decoding steps can be consolidated into a single parallel verification block performed by the server-side \ac{LLM}.
This paradigm shift substantially improves token-generation efficiency by minimizing the number of costly, memory-bound forward passes required at the server.

Despite its computational advantages, deploying \ac{SPIN} at the wireless edge introduces a new communication bottleneck.
This bottleneck emerges not during the initial draft verification, but when a drafted token is rejected and a correction token must be generated.
During verification, acceptance checking only compares the draft probability assigned by the \ac{SLM} to each drafted token against the corresponding target probability assigned by the \ac{LLM},
without requiring the full draft or target distributions~\cite{SDgoogle2023icml, SDdeepmind2023}.
Consequently, the device only needs to upload the drafted tokens and their associated scalar draft probabilities to the server.
This incurs negligible communication overhead.
If all drafted tokens are accepted, the server simply returns lightweight feedback to continue generation.
However, when a drafted token is rejected, a correction token must be sampled from a residual distribution.
This residual distribution inherently depends on the full, vocabulary-wide output distributions of both the \ac{SLM} and the \ac{LLM}~\cite{SDgoogle2023icml, SDdeepmind2023}.
Because modern \acp{LLM} feature massive vocabularies (e.g., 152{,}064 entries for Qwen2.5~\cite{qwen2024qwen25_14b_config} and 200{,}019 entries for GPT-4o~\cite{openai2024cookbook_tiktoken}), exchanging these high-dimensional distributions across the wireless link to enable token correction incurs substantial communication overhead, termed \emph{token-correction overhead}.

Existing studies~\cite{cezheng2025wcsp, tony2025icmlcn, guangyizhang2026quantize, bhattacharjee2025conformal, zheng2026multispin} mainly support the correction step in an \ac{UL} manner: the device uploads the draft distribution to the server, where the corresponding target distribution is already available.
We refer to this implementation as the \emph{\ac{UL} mode of \ac{SPIN}}.
To alleviate the resulting \ac{UL} burden, prior works have explored uncertainty estimation to skip unnecessary \ac{UL} transmissions~\cite{tony2025icmlcn}, as well as quantization~\cite{guangyizhang2026quantize} and sparsification~\cite{bhattacharjee2025conformal} techniques to compress the uploaded draft distributions.
Despite these advances, draft-distribution uploading can still dominate the \ac{E2E} latency due to the limited \ac{UL} rates of edge devices, thereby offsetting the acceleration gain of \ac{SPIN}.
In fact, sampling the correction token does not necessarily have to be performed at the server.
Since the draft distribution is already available at the device,
the correction step can alternatively be carried out at the device by downloading the target distribution from the server~\cite{cezheng2025icml_wksp, jia2026covspec}.
We refer to this \ac{DL}-based implementation as the \emph{\ac{DL} mode of \ac{SPIN}}.
By shifting full-vocabulary distribution transmission from the \ac{UL} to the \ac{DL}, the \ac{DL} mode can exploit the typically stronger \ac{DL} capability of wireless edge systems and thereby bypass the \ac{UL} bottleneck.
In the point-to-point settings considered in existing studies, switching from the \ac{UL} to the \ac{DL} mode is straightforward. 
In a multi-user system, however, this \ac{DL} advantage is no longer unconditional, as discussed below.

In this work, we study \ac{Multi-SPIN} to enable cooperative token generation in a multi-device edge network.
In this architecture, each device employs a local \ac{SLM} to draft tokens for its specific generation task, while the edge server batches the drafted token sequences from multiple devices for parallel verification using a shared \ac{LLM}.
Despite high \ac{DL} rates, assigning all devices to the \ac{DL} mode can severely overload the shared \ac{DL} resources and eventually degrade the overall token-generation efficiency.
This reveals a fundamental tension between the limited \ac{UL} rates and the shared \ac{DL} resource constraints.
To resolve this tension, devices must adaptively switch between the \ac{UL} and \ac{DL} modes to optimally exploit heterogeneous wireless resources.
This motivates us to introduce (communication) \emph{mode selection} as a new design dimension for \ac{Multi-SPIN}.
Indeed, mode selection extends the core resource-utilization philosophy of \ac{Multi-SPIN} from the computation domain to the communication domain.
Much like offloading draft generation to the devices harnesses distributed computing power, mode selection alleviates the shared \ac{DL} bottleneck by adaptively shifting part of the transmission load to the \ac{UL}. 
This perspective highlights mode selection as a natural and complementary mechanism for maximizing the \ac{E2E} efficiency of \ac{Multi-SPIN}.

The choice of communication mode is inherently coupled with the control of draft lengths, which directly affects the token-generation efficiency. 
Specifically, the draft length determines the maximum number of tokens that can be accepted in a single verification round, the local drafting and server-side verification delays, and the probability of a draft rejection.
This creates a fundamental tradeoff: while a longer draft enhances the potential speculative gain by exposing more tokens to parallel verification, it simultaneously increases the likelihood of rejection, thereby triggering the need for costly distribution transmissions.
Consequently, the draft lengths must be determined jointly with the mode-selection mechanism, which is itself tightly intertwined with transmit-power allocation. 
These complex interdependencies necessitate a joint design of draft-length control, mode selection, and power allocation to maximize the sum token goodput, defined as the expected number of output/committed tokens per second across all devices. 
The introduction of this communication-mode-selection framework, together with the formulation and solution of the resulting joint design problem, constitutes the core contribution of this work, as summarized below.
\begin{itemize}
    \item \textbf{\ac{UL}/\ac{DL} Mode Selection for \ac{Multi-SPIN}:}
    We propose a mode-selection-based \ac{Multi-SPIN} protocol that enables each device with a rejected draft to perform correction via either the \ac{UL} or \ac{DL} mode.
    The resulting design problem is inherently stochastic because draft lengths must be selected at the beginning of each round, before the random verification outcomes are revealed.
    After the rejected devices are identified, their communication modes and transmit powers are determined accordingly.
    We formulate the joint design of draft-length control, mode selection, and power allocation as a two-stage stochastic optimization problem for maximizing the expected sum token goodput.
    By exploiting the sequential decision structure, without loss of optimality, we decompose the problem into i) a subproblem of reactive mode selection and power allocation for each realized verification outcome, and ii) a subproblem of proactive draft-length control that incorporates the optimized reactive decisions across all possible outcomes.
    
    \item \textbf{Optimal Reactive Mode Selection and Power Allocation:}
    We characterize the globally optimal reactive communication strategy for any realized verification outcome.
    By sorting the devices according to their minimum \ac{UL} transmission delays, we prove that an optimal mode assignment has a prefix structure: devices with smaller \ac{UL} delays operate in the \ac{UL} mode, while the remaining devices operate in the \ac{DL} mode.
    This property reduces the original combinatorial optimization to a one-dimensional bisection search over the \ac{UL}/\ac{DL} split point.
    For each candidate mode split point, the optimal transmit powers are determined by maximum-power \ac{UL} transmission and delay-equalizing \ac{DL} power allocation.
    This yields the optimal low-complexity policy for joint mode selection and power allocation.
    
    \item \textbf{Proactive Draft-Length Control:}
    We develop a greedy-search-based algorithm that optimizes device-specific draft lengths under heterogeneous computation and communication conditions. 
   	Specifically, the design accounts for device-side computation heterogeneity, rejection uncertainty, and the communication latency under the optimal reactive mode-selection and power-allocation decisions.
    To efficiently evaluate the expected token goodput, we employ a \ac{SAA} approach that avoids exhaustive enumeration of the exponentially many verification outcomes.
    
    \item \textbf{Experimental Results:}
    We conduct extensive experiments using both Qwen2.5~\cite{qwen2024qwen25_14b_config}  and DeepSeek-R1~\cite{guo2025deepseek} model pairs across diverse code-generation, mathematical-reasoning, conversational, and instruction-following tasks.
    The results show that the proposed framework consistently outperforms fixed-mode and conventional draft-length baselines by adaptively balancing \ac{UL}/\ac{DL} communication and tailoring draft lengths to heterogeneous system conditions.
\end{itemize}

The remainder of this paper is organized as follows.
Section~\ref{sec:spin_model_edge} reviews the \ac{SPIN} mechanism and discusses key insights for edge deployments.
Section~\ref{sec:multi_spin_protocol} presents the mode-selection-based \ac{Multi-SPIN} protocol and formulates the two-stage token-goodput maximization problem.
Section~\ref{sec:reactive_mode_power} develops the optimal reactive mode-selection and power-allocation solution.
Section~\ref{sec:proactive_draft_length} proposes the proactive draft-length control algorithm.
Section~\ref{sec:experiments} provides experimental results.
Finally, we conclude this paper in Section~\ref{sec:conclusions}.

\section{\ac{SPIN} Mechanism and Edge Deployment Insights}\label{sec:spin_model_edge}
In this section, we first review the core mechanism of \ac{SPIN}~\cite{SDgoogle2023icml, SDdeepmind2023}.
Based on this review, we then discuss key design insights for edge deployment.

\subsection{\ac{SPIN} Mechanism} \label{sec:spin_model}
Fig.~\ref{fig:speculative_decoding} illustrates the draft-then-verify workflow of \ac{SPIN}.
The \ac{SLM} serves as a lightweight proposal model for fast local drafting, whereas the \ac{LLM} serves as the target model whose output distribution should be preserved.
In particular, \ac{SPIN} accelerates \ac{LLM} inference by allowing the \ac{SLM} to draft multiple tokens autoregressively and then using the \ac{LLM} to verify the drafted block in parallel.
Since the entire drafted block is available before verification, the \ac{LLM} can evaluate all drafted positions through a single forward pass. 
The main steps of \ac{SPIN} are described as follows. 

\begin{figure}[t]
	\centering
	\includegraphics[width=.95\columnwidth]{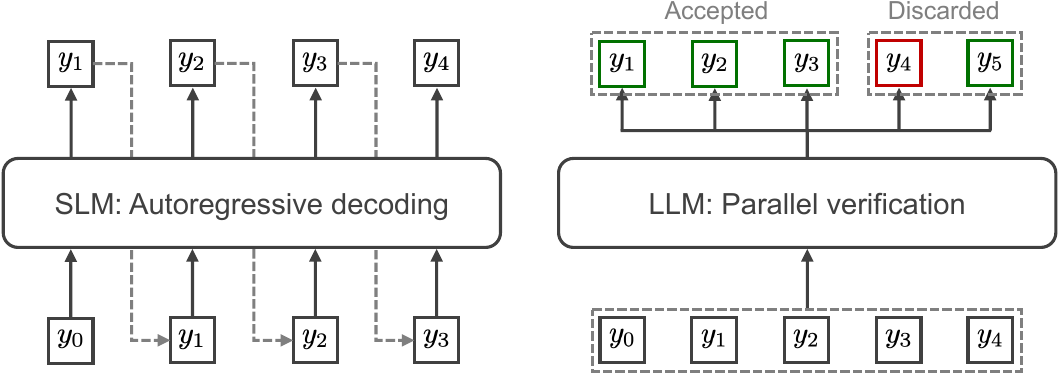}
	 \vspace{-.5em}
	\caption{Autoregressive drafting and parallel verification in \ac{SPIN}.}
	 \vspace{-.5em}
	\label{fig:speculative_decoding}
\end{figure}
\begin{enumerate}
    \item \textbf{Autoregressive drafting:} Given the current prefix $\mbf{x}$, the \ac{SLM} autoregressively generates a block of $\gamma$ draft tokens, denoted by $\mbf{y} \triangleq ( y_{1}, \dots, y_{\gamma})$.
    Specifically, at the $j$-th drafting step, 
    the \ac{SLM} takes $\mbf{x}$ and $\mbf{y}_{<j} \triangleq (y_{1}, \dots, y_{j-1})$ as input and performs one forward pass to obtain the draft distribution $\pi^{\rm S}_j(\cdot \mid \mbf{x},\mbf{y}_{<j})$.
    The draft token $y_j$ is then sampled from this distribution as
    \begin{align}
        y_{j} \sim  \pi^{\rm S}_j (\cdot \mid \mbf{x}, \mbf{y}_{<j}),  ~~j = 1, \dots, \gamma. 
    \end{align}
            
    \item \textbf{Parallel verification:} Given the current prefix $\mbf{x}$ and the drafted block $\mbf{y}$, the \ac{LLM} performs a single forward pass to obtain the target distributions for all drafted positions and the next position, given by
    \begin{align}
        \pi^{\rm L}_j (\cdot \mid \mbf{x}, \mbf{y}_{ <j}), ~~ j = 1, \dots, \gamma +1.
    \end{align}
    After obtaining these target distributions, \ac{SPIN} performs acceptance checking for the drafted tokens.
    Specifically, for the $j$-th drafted token $y_j$, \ac{SPIN} compares the target probability assigned by the \ac{LLM}, denoted by $\pi^{\rm L}_j(y_j) \triangleq \pi^{\rm L}_j(y_j \mid \mbf{x}, \mbf{y}_{<j})$, with the draft probability assigned by the \ac{SLM}, denoted by $\pi^{\rm S}_j(y_j) \triangleq \pi^{\rm S}_j(y_j \mid \mbf{x}, \mbf{y}_{<j})$.
    The token $y_j$ is accepted if
    \begin{align} \label{eqn:acceptance_test}
        u_{j} < \min \left(1,\frac{\pi^{\rm L}_{j}(y_{j})}{\pi^{\rm S}_{j}(y_{j})}\right), ~~ u_{j} \sim \mathrm{Unif}(0,1),
    \end{align}
    where $\mathrm{Unif}(0,1)$ denotes the uniform distribution over the interval $(0,1)$.

    \item \textbf{Bonus-/correction-token sampling:}
    Once the verification outcomes are determined, \ac{SPIN} proceeds according to the following two cases:
    \begin{itemize}
        \item \textbf{Case I (all drafted tokens accepted):}
        If all $\gamma$ drafted tokens are accepted, 
        the next-position target distribution $\pi^{\rm L}_{\gamma+1}(\cdot) \triangleq \pi^{\rm L}_{\gamma+1}(\cdot \mid \mbf{x},\mbf{y})$ is used to sample one additional bonus token as
        \begin{align}
            y_{\gamma + 1} \sim \pi^{\rm L}_{\gamma + 1}(\cdot).
        \end{align}
        The accepted draft tokens and the bonus token are appended to the prefix.
        The current round then terminates, and the next round starts from the updated prefix.

        \item \textbf{Case II (rejection occurs):}
        Otherwise, let $r$ denote the position of the first rejected token.
        Then, all drafted tokens before position $r$ are accepted.
        The drafted token at position $r$ is replaced by a correction token sampled from the residual distribution.
        Specifically,
        \begin{align} \label{eqn:resample}
            y^\prime_{r} \sim \mathrm{norm}\left(\max\left(0,\, \pi^{\rm L}_{r}(\cdot)-\pi^{\rm S}_{r}(\cdot)\right)\right),
        \end{align}
        where $\mathrm{norm}(\cdot)$ denotes the normalization operator that rescales a nonnegative function over the vocabulary into a valid probability distribution.
        All drafted tokens after position $r$ are discarded.
        The accepted draft tokens before $r$ and the correction token $y^\prime_{r}$ are appended to the prefix before the next round starts.
    \end{itemize}
\end{enumerate}

In this way, \ac{SPIN} preserves the exact target distribution of the \ac{LLM} while allowing multiple tokens to be generated in one round.

\subsection{Key Insights for Edge Deployment}
We next identify the information required at different steps of \ac{SPIN}, which provides the basis for characterizing the communication overhead of deploying \ac{SPIN} at the wireless edge.

\begin{observ}
    In parallel verification, acceptance checking requires only the scalar draft and target probabilities assigned to each drafted token.
    If all drafted tokens are accepted, bonus-token sampling requires the next-position target distribution from the \ac{LLM}, without any additional information from the \ac{SLM}.
    By contrast, when a drafted token is rejected, correction-token sampling requires access to both the full draft and target distributions at the rejection position.
\end{observ}

The above observation suggests that, when all drafted tokens are accepted, the communication overhead is negligible.
In this case, the device only needs to upload the drafted tokens and their corresponding scalar draft probabilities, while the server only needs to return the bonus token.
The critical communication bottleneck arises upon token rejection, since correcting the rejected token requires transmitting a vocabulary-wide distribution.

\begin{observ}
    Since the draft distribution is available at the device while the target distribution is available at the server, correction-token sampling can be performed either at the server via draft-distribution uploading or at the device via target-distribution downloading.
\end{observ}

We refer to the above two reciprocal implementations as the \ac{UL} mode and the \ac{DL} mode of \ac{SPIN}, respectively.
In practical cellular systems, the \ac{DL} rate is often higher than the \ac{UL} rate because the server typically has a larger transmit power budget than edge devices.
This makes the \ac{DL} mode more attractive for reducing the distribution-transmission delay.
However, this \ac{DL} advantage does not directly extend to multi-access systems.
When multiple devices require correction-token sampling in the same round, assigning all of them to the \ac{DL} mode forces them to compete for the server’s shared \ac{DL} transmit power, which can substantially reduce the per-device \ac{DL} rate.
In this case, shifting a suitable subset of devices to the \ac{UL} mode can better exploit available device-side \ac{UL} capabilities and significantly relieve the shared \ac{DL} burden.
Therefore, adaptive \ac{UL}/\ac{DL} mode selection becomes essential for balancing link-rate asymmetry against shared-resource competition, which motivates the work presented in this paper.

\begin{figure}[t]
	\centering
	\includegraphics[width=1\columnwidth]{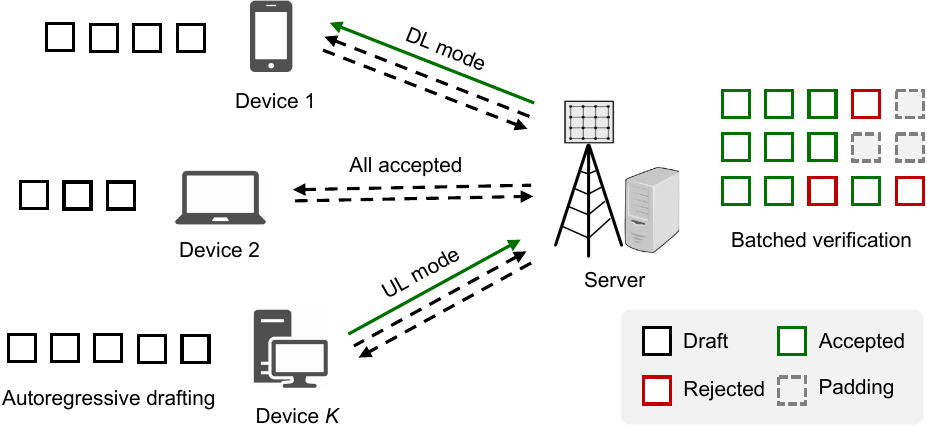}
	\caption{A \ac{Multi-SPIN} system with communication-mode selection.}
	\label{fig:system_model}
\end{figure}

\section{Communication-Mode Selection for Multi-SPIN}\label{sec:multi_spin_protocol}
As shown in Fig.~\ref{fig:system_model}, we consider a Multi-\ac{SPIN} system for \ac{LLM} inference services, where $K$ edge devices perform independent token generation tasks with the assistance of an edge server.
Specifically, each device sequentially generates multiple draft tokens using its local \ac{SLM}, and the edge server verifies them in parallel using an \ac{LLM}.
This design exploits both intra-request and inter-request parallelism: the server verifies multiple drafted tokens within each request in parallel, while simultaneously batching requests across devices~\cite{su2023synergy,zhang2025batch_done_right,xidian2025batch}.
In this way, multiple tokens across multiple devices can be accepted in a single \ac{LLM} forward pass.
After the verification step, any device whose draft is rejected must perform an additional correction procedure before it can proceed.
Such correction can be supported by exchanging the required distribution information in either the \ac{UL} mode, where the draft distribution is uploaded to the server, or the \ac{DL} mode, where the target distribution is downloaded to the device.
This flexibility reveals mode selection as a unique design dimension in \ac{Multi-SPIN}.

\subsection{Mode-Selection-Based Multi-\ac{SPIN} Protocol}\label{subsec:multi_spin_protocol}
In this subsection, we develop the mode-selection-based \ac{Multi-SPIN} protocol,
which comprises a set of common steps followed by outcome-dependent operations.
\subsubsection{Common Steps}
At the beginning of each token generation round, the server determines the draft length $\gamma_k$ for each device and informs the devices accordingly.
The protocol then proceeds through the following common steps:
\begin{enumerate}
    \item \textbf{On-device autoregressive drafting:}
    For each device $k \in \mathcal{K} \triangleq \{1, \dots, K\}$, given the current prefix $\mbf{x}_k$, the local \ac{SLM} autoregressively generates a block of $\gamma_k$ draft tokens, denoted by $\mbf{y}_k \triangleq ( y_{k,1}, \dots, y_{k, \gamma_k})$.
    Specifically, the $j$-th token is sampled from the draft distribution as
    \begin{align}
        \!\!\!y_{k,j} \sim \pi^{\rm S}_{k,j}  (\cdot) \triangleq \pi^{\rm S}_{k,j}  (\cdot \mid \mbf{x}_k, \mbf{y}_{k,<j}),  ~~j = 1, \dots, \gamma_k, 
    \end{align}
    where $\mbf{y}_{k, <j} \triangleq (y_{k,1}, \dots, y_{k, j-1})$.
    
    \item \textbf{Drafts uploading:}
    Device $k$ uploads the drafted token--probability pairs $\{y_{k,j}, \pi^{\rm S}_{k,j}(y_{k,j})\}_{j=1}^{\gamma_k}$ to the server.
    
    \item \textbf{Server-side batched parallel verification:}
    Given the current prefixes $\{\mbf{x}_k\}_{k\in\mathcal{K}}$ and the drafted blocks $\{\mbf{y}_k\}_{k\in\mathcal{K}}$ from all devices, the edge server uses the \ac{LLM} to evaluate the drafted blocks in parallel across both devices and drafted positions. 
    Specifically, for each device $k\in\mathcal K$, the \ac{LLM} outputs the target distributions for all drafted positions and the next position as
    \begin{align}
        \pi^{\rm L}_{k,j} (\cdot) \triangleq \pi^{\rm L}_{k,j} (\cdot \mid \mbf{x}_k, \mbf{y}_{k,<j}),
        ~~ j = 1,\dots,\gamma_k +1.
    \end{align}
    For each drafted token $y_{k,j}$, the server accepts it if
    \begin{align}
        u_{k,j} < \min \left(1,\frac{\pi^{\rm L}_{k,j}(y_{k,j})}{\pi^{\rm S}_{k,j}(y_{k,j})}\right),
        ~~ u_{k,j} \sim \mathrm{Unif}(0,1).
    \end{align}
\end{enumerate}

\begin{figure}[t]
	\centering
	\includegraphics[width=.66\linewidth]{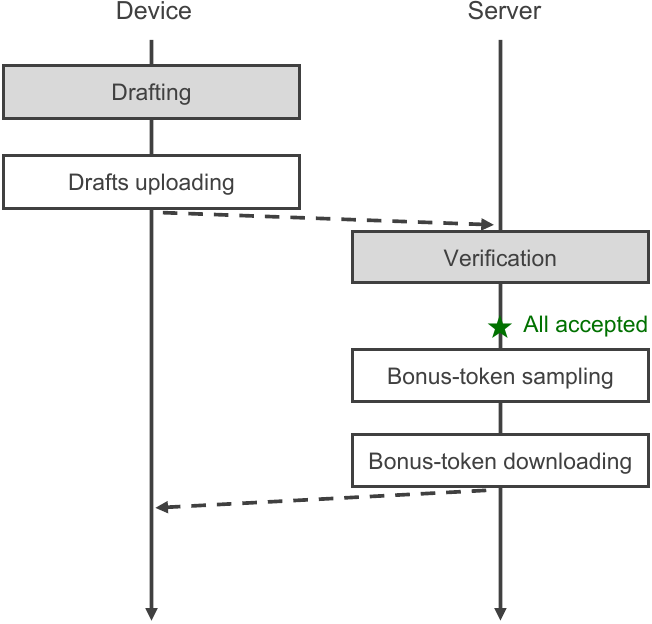}
	\caption{Protocol flow when all drafted tokens from a device are accepted.}
	\label{fig:all_accepted}
\end{figure}

Based on these token-wise acceptance decisions, we define the verification outcome for device $k$ as
\begin{align}
    o_k = \begin{cases}
        0, & \text{if all drafted tokens are accepted},\\
        1, & \text{otherwise},
    \end{cases}
\end{align}
Collecting the device-wise outcomes yields $\mbf{o} \triangleq \left[o_1, \dots, o_K\right]^\top$.
The corresponding set of devices with rejected drafts is then defined as $\mathcal{R}(\mbf{o}) \triangleq \left\{k \in \mathcal{K} \mid o_k = 1\right\}$.
For a realized verification outcome, we write $\mathcal{R}$ in place of $\mathcal{R}(\mbf{o})$ when no ambiguity arises.
The subsequent operations are described separately for devices with all drafted tokens accepted, i.e., $k \in \mathcal{K}\setminus\mathcal{R}$, and devices with at least one rejected token, i.e., $k \in \mathcal{R}$.

\subsubsection{Devices with All Drafted Tokens Accepted}
For each device $k \in \mathcal{K} \setminus \mathcal{R}$, the protocol proceeds as illustrated in Fig.~\ref{fig:all_accepted} and described below:
\begin{enumerate}[label=\arabic*A), start=4]
    \item \textbf{Bonus-token sampling:}
    The server samples an additional bonus token from the target distribution as
    \begin{align}
        y_{k, \gamma_k+1} \sim \pi^\mathrm{L}_{k, \gamma_k +1}(\cdot ).
    \end{align}
    \item \textbf{Bonus-token downloading:} 
    The server sends the bonus token $y_{k, \gamma_k+1}$ to device $k$.
    The device and server append the accepted drafted tokens and the bonus token to their respective copies of the prefix.
     The current round then terminates, and the next round starts from the updated prefix.
\end{enumerate}

\subsubsection{Devices with Rejected Drafts}
We next describe the correction procedure for devices with rejected drafts.
For each device $k\in\mathcal{R}$, the server selects its communication mode $m_k$ and associated transmit power $p_k$,
where 
\begin{align}
    m_k =
    \begin{cases}
        0, & \text{if device $k$ operates in the \ac{UL} mode},\\
        1, & \text{if device $k$ operates in the \ac{DL} mode}.
    \end{cases}
\end{align}
Here, $p_k$ denotes the \ac{UL} transmit power of device $k$ when $m_k=0$, and the \ac{DL} transmit power allocated by the server to device $k$ when $m_k=1$.

\begin{figure}[t]
	\centering
	\vspace{-.5em}
	\hspace*{0.02\linewidth}
	\subfigure[\ac{UL} mode]{
		\includegraphics[width=.7\linewidth]{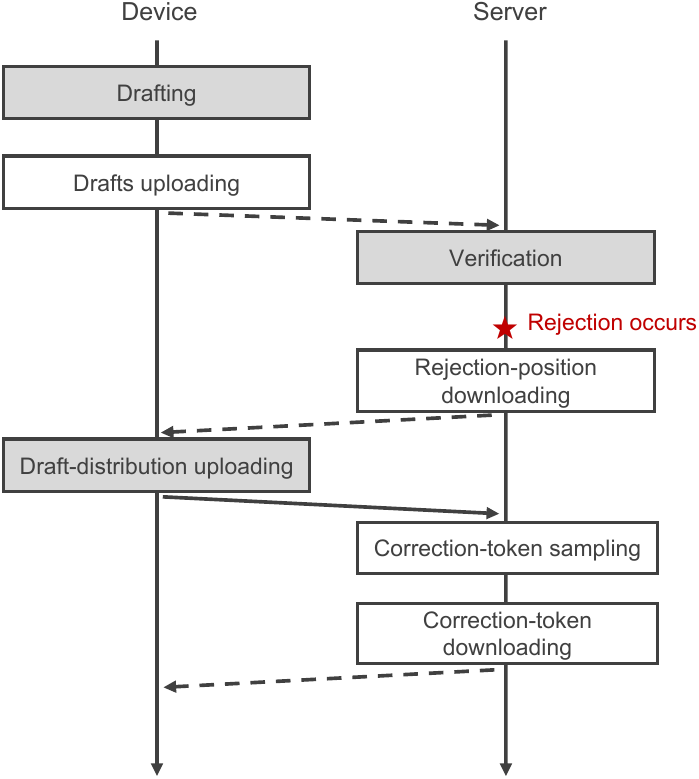}
		\label{fig:protocol_ul}
	}\\[.5em]
	\hspace*{0.02\linewidth}
	\subfigure[\ac{DL} mode]{
		\includegraphics[width=.7\linewidth]{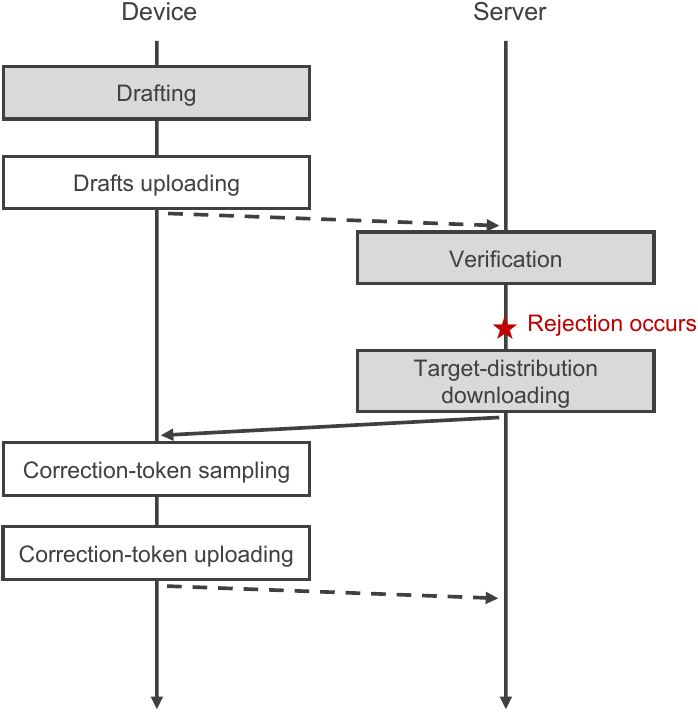}
		\label{fig:protocol_dl}
	}
	\caption{Protocol flows when a drafted token from a device is rejected, with correction performed in the (a) \ac{UL} mode and (b) \ac{DL} mode.}
	\label{fig:protocol}
\end{figure}

As illustrated in Fig.~\ref{fig:protocol_ul}, in the \ac{UL} mode, correction is performed at the server after the device uploads its draft distribution.
Devices assigned to this mode proceed as follows:
\begin{enumerate}[label=\arabic*U), start=4]
    \item \textbf{Rejection-position downloading:}
    The server sends the index $r_k$ of the first rejected token to device $k$.
    
    \item \textbf{Draft-distribution uploading:} 
    Device $k$ sends the draft distribution $\pi^{\rm S}_{k, r_k}(\cdot)$ at position $r_k$ to the server via the \ac{UL}.
    
    \item \textbf{Server-side correction:}
    The server samples the correction token at position $r_k$ from the residual distribution as
    \begin{align} \label{eqn:residual_distribution}
        y^\prime_{k, r_k} \sim \mathrm{norm}\left(\max\left(0,\, \pi^{\rm L}_{k, r_k}(\cdot)-\pi^{\rm S}_{k, r_k}(\cdot)\right)\right).
    \end{align}
    
    \item \textbf{Correction-token downloading:}
    The server sends the correction token $y^\prime_{k, r_k}$ to device $k$.
    The device and server append the accepted drafted tokens before position $r_k$ and the correction token to their respective copies of the prefix.
     The current round then terminates, and the next round starts from the updated prefix.
\end{enumerate}

As illustrated in Fig.~\ref{fig:protocol_dl}, in the \ac{DL} mode, correction is performed at the device after it receives the target distribution from the server.
Devices assigned to this mode proceed as follows:
\begin{enumerate}[label=\arabic*D), start=4]
    \item \textbf{Target-distribution downloading:}
    The server sends the first rejection position $r_k$ and the corresponding target distribution $\pi^{\rm L}_{k,r_k}(\cdot)$ to device $k$ via the \ac{DL}.
    
    \item \textbf{On-device correction:}
    Device $k$ samples the correction token at position $r_k$ from the residual distribution in~\eqref{eqn:residual_distribution}.
    
    \item \textbf{Correction-token uploading:}
    Device $k$ sends the correction token $y^\prime_{k, r_k}$ to the server.
    The device and server append the accepted drafted tokens before position $r_k$ and the correction token to their respective copies of the prefix.
     The current round then terminates, and the next round starts from the updated prefix.
\end{enumerate}

\begin{figure}[t]
	\centering	\includegraphics[width=.85\columnwidth]{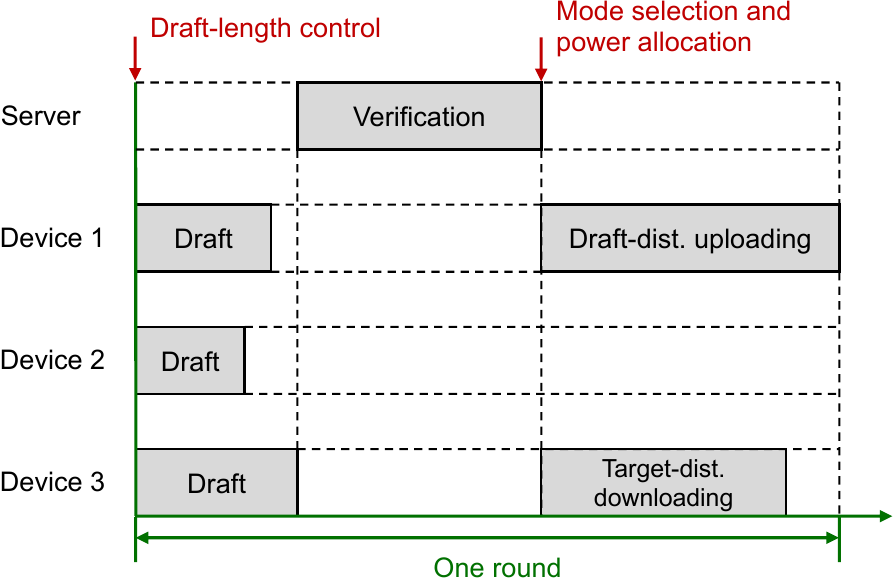}
    \vspace{-.3em}
	\caption{Timeline of one synchronized round in the mode-selection-based \ac{Multi-SPIN} protocol.
    For clarity, lightweight signaling steps are omitted.
    }
	\label{fig:multi_spin_protocol}
\end{figure}

Fig.~\ref{fig:multi_spin_protocol} illustrates the round-level timing of the proposed protocol under a synchronized implementation. 
At the beginning of each round, the server proactively selects the draft length $\gamma_k$ for each device $k$ and informs the devices accordingly.
The devices then perform local drafting, while the server waits until all drafted token--probability pairs have been received before conducting batched parallel verification.
After the verification outcomes are revealed, only devices with rejected drafts participate in distribution transmission, and their communication modes $m_k$ and transmit powers $p_k$ are determined reactively. 
For clarity, Fig.~\ref{fig:multi_spin_protocol} omits the lightweight signaling steps, such as drafts uploading, rejection-position feedback, bonus-token downloading, and correction-token exchange, whose overhead is negligible compared with vocabulary-wide distribution transmission. 
This proactive–reactive timing structure forms the basis of the two-stage optimization developed later.

\subsection{Computation Throughput and Latency}
\subsubsection{Expected Number of Output Tokens}
We denote the draft-length vector by $\bsm{\gamma} \triangleq [\gamma_1,\ldots,\gamma_K]^\top$.
Let $L_k(\gamma_k)$ denote the number of output/committed tokens generated for device $k$ in one round, including the accepted drafted tokens and one additional bonus or correction token.
The total number of output tokens across all devices is 
\begin{align}
    L(\bsm{\gamma}) \triangleq \sum_{k\in\mathcal{K}} L_{k}(\gamma_k).
\end{align}
Following~\cite{SDgoogle2023icml, SDdeepmind2023}, we assume that, for each device $k$, the acceptance events of its drafted tokens are \ac{i.i.d.} with acceptance probability $\alpha_k$, and are independent across devices.
Under this model, the expected number of output tokens in one \ac{Multi-SPIN} round is given by~\cite{SDgoogle2023icml}
\begin{align}\label{eqn:expected_output_tokens}
    \bar{L}(\bsm{\gamma}) \triangleq \mathbb{E}[L(\bsm{\gamma})] = \sum_{k\in\mathcal{K}} \mathbb{E}[L_{k}(\gamma_k)]
    = \sum_{k\in\mathcal{K}} \frac{1 - \alpha_k^{\gamma_k+1}}{1 - \alpha_k}.
\end{align}

\subsubsection{Computation Latency}
We characterize the computation latency of local drafting and server-side verification as follows.
Device $k$ autoregressively generates $\gamma_k$ draft tokens, incurring a drafting latency of $\gamma_k T_k^{\mathrm{dft}}$, where $T_k^{\mathrm{dft}}$ is the average per-token generation latency.
Since the server waits for all devices before verification, the overall drafting latency is
\begin{align}
    T^{\mathrm{dft}}(\bsm{\gamma}) = \max_{k \in \mathcal{K}}~ \left\{\gamma_k T_k^{\mathrm{dft}} \right\}.
\end{align}

Verification is then performed by the server-side \ac{LLM} in a batched manner.
As shown in Fig.~\ref{fig:system_model}, the \ac{LLM} forward pass exploits parallelism along two dimensions: draft positions within each request and requests across devices.
For the moderate draft lengths commonly used in \ac{SPIN}, the verification latency depends only weakly on the draft length and is primarily affected by the batch size, i.e., the number of device requests verified together~\cite{su2023synergy,zhang2025batch_done_right,xidian2025batch}.
Following~\cite{zheng2026multispin}, we model the batched verification latency as
\begin{align}\label{eqn:verification_latency}
    T^{\mathrm{ver}} = T_0^{\mathrm{ver}} + KT_1^{\mathrm{ver}},
\end{align}
where $T_0^{\mathrm{ver}}$ captures the fixed latency of one batched \ac{LLM} verification pass, and $T_1^{\mathrm{ver}}$ captures the incremental latency introduced by each additional device in the batch.

The total computation latency in one round is therefore
\begin{align}
    T^\mathrm{comp}(\bsm{\gamma}) = T^\mathrm{dft} (\bsm{\gamma}) + T^\mathrm{ver}.
\end{align}

\subsection{Communication Overhead and Latency}
\subsubsection{Communication Overhead}
As discussed earlier, the distribution transmission required for correction-token sampling constitutes the dominant communication overhead in \ac{Multi-SPIN}.
We therefore focus on this cost and neglect the signaling overhead of other protocol operations.

Upon a rejection, the selected communication mode determines whether the draft distribution or the target distribution is transmitted.
Since both distributions are defined over the same vocabulary, transmitting either distribution incurs the same overhead.
For typical \acp{LLM}, the vocabulary size is on the order of tens of thousands or even larger; for example, Qwen2.5 has 152{,}064 vocabulary entries~\cite{qwen2024qwen25_14b_config}, while GPT-4o has 200{,}019 vocabulary entries~\cite{openai2024cookbook_tiktoken}.
In practice, since many tokens have near-zero probabilities~\cite{guangyizhang2026quantize}, token distributions are effectively sparse after truncation.
Hence, it suffices to transmit only the effective support of the distribution, whose size is typically on the order of hundreds or thousands.
Let $V$ denote the vocabulary size, $V^{\mathrm{eff}}$ the effective support size, and $Q$ the number of quantization bits used for each probability value.
Each transmitted entry contains a token index and its quantized probability, requiring $\lceil \log_2 V\rceil + Q$ bits.
Accordingly, for each device with a rejected draft, the distribution-transmission overhead in one round is $V^{\mathrm{eff}} \left(\lceil \log_2 V\rceil + Q\right)$ bits.

\subsubsection{Communication Latency}
We consider an \ac{OFDMA} system in which each device $k \in \mathcal{R}$ is allocated an orthogonal bandwidth $B_k$.
The allocated bandwidth is used for either \ac{UL} or \ac{DL} distribution transmission according to the selected communication mode.
For simplicity, the bandwidth allocation is treated as given.
Each device and the server are equipped with a single antenna.
Let $h_k$ denote the channel coefficient between device $k$ and the server.
We adopt a block-fading channel model, under which $h_k$ remains constant within each \ac{Multi-SPIN} round and varies independently across rounds.
The channel coefficients are assumed known at the server at the beginning of each round, before the draft-length decision is made.
Accordingly, the optimization in the sequel is conditioned on the current channel realization, whose dependence is suppressed for notational simplicity.
The achievable communication rate of device $k$ is
\begin{align}
    R_{k} (m_k, p_k) = \begin{cases}
        B_k \log_2 \left(1+\frac{p_{k} |h_k|^2}{N_0^\mathrm{ser} B_k}\right), & m_k = 0 \\
        B_k \log_2 \left(1+\frac{p_{k} |h_k|^2}{N_{0, k}^\mathrm{dev} B_k}\right), & m_k = 1,
    \end{cases}
\end{align}
where $N_0^\mathrm{ser}$ and $N_{0,k}^\mathrm{dev}$ denote the noise power spectral density at the server and device $k$, respectively.
	
For each device $k \in \mathcal{R}$, the latency for transmitting the required distribution is
\begin{align}
	T_k^\mathrm{comm} (m_k,p_k) = 
	\frac{V^{\rm eff}(\lceil \log_2 V\rceil + Q)}{R_k(m_k, p_k)}.
\end{align}
Since the distribution transmissions of the devices in $ \mathcal{R}$ are performed simultaneously over orthogonal subcarriers, the overall communication latency is determined by the slowest device, i.e.,
\begin{align}
	T^\mathrm{comm} (\mathsf{A}; \mbf{o}) = \max_{k \in \mathcal{R}(\mbf{o})} ~T_k^\mathrm{comm} (m_k, p_k),
\end{align}
where $\mathsf{A} \triangleq \left(\mbf{m}_{\mathcal{R}}, \mbf{p}_{\mathcal{R}}\right)$ denotes the reactive communication action, with $\mbf{m}_{\mathcal{R}}\triangleq[m_k]_{k\in\mathcal{R}}$ and $\mbf{p}_{\mathcal{R}}\triangleq[p_k]_{k\in\mathcal{R}}$ collecting the mode-selection and power-allocation decisions for the devices with rejected drafts, respectively.
When the rejection set is clear from the context, we omit the subscript $\mathcal R$ and simply write $\mathbf m$ and $\mathbf p$.
When $\mathcal R(\mathbf o)=\emptyset$, all devices have their drafted tokens accepted and no distribution transmission is required.
Accordingly, we define $T^{\mathrm{comm}}(\mathsf A;\mathbf o)=0$.

\subsection{Token-Goodput Maximization Problem}
We now formulate the joint optimization of communication modes, draft lengths, and transmit powers to maximize the expected sum token goodput.
As described in Section~\ref{subsec:multi_spin_protocol}, these decisions are made at different stages of each generation round.
Specifically, the draft lengths are selected proactively at the beginning of each round before the verification outcomes are known.
After the devices with rejected drafts are identified, their communication modes and transmit powers are determined reactively.
This sequential decision process naturally leads to a two-stage proactive--reactive optimization problem.

For a given draft-length vector $\boldsymbol{\gamma}$, verification outcome $\mbf{o}$, and reactive action $ \mathsf{A}$, the latency of one \ac{Multi-SPIN} round is
\begin{align}
	T(\bsm{\gamma}, \mathsf{A}; \mbf{o}) = T^\mathrm{comp}(\bsm{\gamma}) + T^\mathrm{comm} (\mathsf{A}; \mbf{o}).
\end{align}
Because $\bsm{\gamma}$ is selected before verification, the round latency remains unknown when the draft-length decision is made.
In particular, the verification outcome $\mathbf{o}$, whose distribution depends on $\boldsymbol{\gamma}$, determines which devices require correction and hence the incurred communication latency.
Therefore, the performance of a draft-length decision should be evaluated in terms of the expected round latency over all possible verification outcomes.
To characterize this expected latency, we define a reactive policy $\chi$ that maps each realized verification outcome to a feasible communication action, i.e., $\mathsf{A} = \chi (\mbf{o})$.
Then, the expected round latency under $\bsm{\gamma}$ and $\chi$ is
\begin{align}
	\bar{T}(\bsm{\gamma}, \chi) 
	&\triangleq \mathbb{E}_{\mbf{o}|\bsm{\gamma}} \! \left[T\big(\bsm{\gamma}, \chi (\mbf{o}); \mbf{o}\big)\right] \nonumber \\
	&= T^\mathrm{comp}(\bsm{\gamma}) + \mathbb{E}_{\mbf{o}|\bsm{\gamma}} \! \left[T^\mathrm{comm} \big(\chi(\mbf{o}); \mbf{o}\big)\right],
\end{align}
where the expectation is taken over the random verification outcomes $\mbf{o}$ induced by $\bsm{\gamma}$.

Combining the expected number of output tokens in~\eqref{eqn:expected_output_tokens} with the expected round latency, we define the expected sum token goodput under the proactive draft-length vector $\boldsymbol{\gamma}$ and the reactive communication policy $\chi$ as
\begin{align}
	G(\bsm{\gamma}, \chi) \triangleq \frac{\bar{L}(\bsm{\gamma})}{\bar{T}(\bsm{\gamma}, \chi)}.
\end{align}

We next specify the feasibility constraints on the proactive and reactive decisions.
The draft length of each device must satisfy
\begin{align}
	\gamma_k \in \mathbb{Z}_{\geq 0},
	\quad
	\forall k \in \mathcal{K},
\end{align}
where $\mathbb{Z}_{\geq 0}$ denotes the set of nonnegative integers.
For each realized verification outcome $\mathbf{o}$, the reactive policy $\chi$ selects a communication action $\chi(\mathbf{o})=(\mathbf{m},\mathbf{p})$ from the feasible action region
\begin{align}
	\mathcal{A}(\mbf{o})
	\triangleq
	\left\{
	(\mbf{m}, \mbf{p})
	\;\middle|
	\begin{array}{ll}
		m_k \in \{0,1\}, 
		& \forall k \in \mathcal{R}(\mbf{o}), \\[1pt]
		p_k \geq 0, 
		& \forall k \in \mathcal{R}(\mbf{o}), \\[1pt]
		(1-m_k)p_k \leq P_k^{\mathrm{ul}}, 
		& \forall k \in \mathcal{R}(\mbf{o}), \\[1pt]
		\sum_{k \in \mathcal{R}} m_k p_k
		\leq P^{\mathrm{dl}}.
		&
	\end{array}
	\right\}.
\end{align}
The third constraint limits the transmit power of device $k$ when it operates in the \ac{UL} mode, whereas the last constraint imposes the shared \ac{DL} power budget at the server.
Accordingly, the set of feasible reactive policies is defined as
\begin{align}
	\mathcal{X} \triangleq \left\{\chi \mid \chi (\mbf{o}) \in \mathcal{A}(\mbf{o}), \forall \mbf{o} \in \{0, 1\}^K \right\}.
\end{align}

The joint optimization of proactive draft lengths and reactive communication decisions is therefore formulated as the following two-stage expected sum-token-goodput maximization problem:
\begin{subequations}
	\begin{alignat}{2}
		\max_{\bsm{\gamma}, \chi}\quad
		& G(\bsm{\gamma},\chi)
		\\
		\operatorname{ s.t. } \quad
		& \gamma_k \in \mathbb{Z}_{\geq 0},
		&\quad&
		\forall k \in \mathcal{K},
		\\
		& \chi \in \mathcal{X}.
		&\quad&
	\end{alignat}
\end{subequations}
The two stages are coupled through the verification outcomes.
Specifically, the proactive draft-length vector $\boldsymbol{\gamma}$ determines the distribution of verification outcomes $\mathbf{o}$, while each realized outcome determines the set of devices requiring correction and hence the corresponding reactive communication problem.
This structure allows the original problem to be decomposed, without loss of optimality, into two subproblems.
The first is a reactive mode-selection and power-allocation subproblem solved for each realized outcome.
The second is a proactive draft-length control subproblem that incorporates the optimized reactive decisions over all possible outcomes.

\subsubsection{Reactive Mode Selection and Power Allocation}
We first derive the reactive communication problem for a predetermined draft-length vector $\boldsymbol{\gamma}$.
In this case, both the expected number of output tokens $\bar{L}(\bsm{\gamma})$ and the computation latency $T^\mathrm{comp}(\bsm{\gamma})$ are independent of the reactive policy $\chi$.
Therefore, maximizing the expected sum token goodput over $\chi$ is equivalent to minimizing the expected communication latency:
\begin{align}
	\min_{\chi\in\mathcal{X}} ~\mathbb{E}_{\mbf{o}|\bsm{\gamma}} \! \left[ T^{\mathrm{comm}} \bigl( \chi(\mbf{o}); \mbf{o} \bigr) \right].
\end{align}
The minimization is separable across verification outcomes, and the optimal reactive action can be determined independently for each realization.
It follows that
\begin{align}
	\min_{\chi\in\mathcal{X}} ~\mathbb{E}_{\mbf{o}|\bsm{\gamma}} \! \left[ T^{\mathrm{comm}} \bigl( \chi(\mbf{o}); \mbf{o} \bigr) \right]
	= 
	\mathbb{E}_{\mbf{o}|\bsm{\gamma}} \! \left[ \min_{\mathsf{A}\in\mathcal{A}(\mbf{o})} T^{\mathrm{comm}} \bigl( \mathsf{A}; \mbf{o} \bigr) \right].
\end{align}
Accordingly, for a realized outcome with $\mathcal R(\mathbf o) \neq \emptyset$, the optimal reactive action can be obtained by solving
\begin{subequations}
	\begin{align}
		T^\mathrm{comm}_{\min}\left(\mbf{o}\right)
		\triangleq \min_{\mbf{m}, \mbf{p}}  \quad
		&\max_{k \in \mathcal{R}(\mbf{o})} ~~ T_{k}^\mathrm{comm} (m_k, p_k) \\ 
		\operatorname{ s.t. }  \quad
		& (\mbf{m}, \mbf{p}) \in \mathcal{A}(\mbf{o}).
	\end{align}
\end{subequations}
When $\mathcal R(\mathbf o)=\emptyset$, we simply define $T_{\min}^{\mathrm{comm}}(\mathbf o)=0$, and no reactive communication decision is required.
The above formulation leads to a mixed-integer optimization problem, in which the binary mode-selection variables $\mathbf{m}$ are jointly optimized with the continuous power-allocation variables $\mathbf{p}$.
The shared \ac{DL} power constraint further couples the decisions of all devices assigned to the \ac{DL} mode.
Moreover, the number of possible mode assignments grows exponentially with the number of devices, making exhaustive search computationally prohibitive.


\subsubsection{Proactive Draft-Length Control}
Under the optimal reactive policy, the communication latency associated with each verification outcome $\mathbf{o}$ is $T_{\min}^{\mathrm{comm}}(\mathbf{o})$.
Hence, for a given draft-length vector $\boldsymbol{\gamma}$, the minimum expected round latency is
\begin{align} \label{eqn:min_expected_round_latency}
	T^{\mathrm{comp}}(\boldsymbol{\gamma})
	+
	\mathbb{E}_{\mbf{o}\mid\boldsymbol{\gamma}}
	\left[
	T_{\min}^{\mathrm{comm}}(\mbf{o})
	\right].
\end{align}
Substituting~\eqref{eqn:min_expected_round_latency} into the joint formulation yields the proactive draft-length control problem:
\begin{subequations}
	\begin{alignat}{2}
		\max_{\bsm{\gamma}}\quad
		& \frac{\bar{L}(\bsm{\gamma})}{T^\mathrm{comp}(\bsm{\gamma}) + \mathbb{E}_{\mbf{o}|\bsm{\gamma}} \! \left[T^\mathrm{comm}_{\min}\left(\mbf{o}\right)\right]}
		\\
		\operatorname{ s.t. } \quad
		& \gamma_k \in \mathbb{Z}_{\geq 0}, \quad \forall k \in \mathcal{K}.
	\end{alignat}
\end{subequations}
This formulation captures the fundamental tradeoff in draft-length control. 
Longer drafts can increase the expected number of output tokens, but they also reshape the verification-outcome distribution and affect both the computation latency and the expected communication latency.
However, solving this stochastic optimization problem is generally challenging.
In particular, direct evaluation of the objective requires averaging over an exponentially large outcome space, and for each outcome, solving the associated reactive communication problem to obtain $T_{\min}^{\mathrm{comm}}(\mathbf{o})$.
This tight coupling renders draft-length optimization intrinsically dependent on the underlying reactive problem.

\section{Reactive Mode Selection and Power Allocation}\label{sec:reactive_mode_power}
In this section, we solve the reactive mode-selection and power-allocation problem under a given verification outcome.
We first derive the optimal power allocation for any fixed mode-selection strategy.
Then, by exploiting structural properties of the optimal mode selection, we reduce the original combinatorial problem to a one-dimensional bisection search, which yields the optimal mode selection and power allocation jointly.

\subsection{Optimal Power Allocation}
For a fixed mode selection,
let $\mathcal{U} \triangleq \left\{k \in \mathcal{R}\mid m_k = 0\right\}$ and $\mathcal{D} \triangleq \left\{k \in \mathcal{R}\mid m_k = 1\right\}$ denote the corresponding \ac{UL} and \ac{DL} user sets, respectively.
We have 
$\mathcal{U} \cup \mathcal{D} = \mathcal{R}$.
For a nonempty \ac{UL} user set $\mathcal U$, each device transmits at its maximum available power to minimize the transmission latency, which is given by
\begin{align}
    \tau^{\mathrm{ul}} = \max_{k \in \mathcal{U}} ~\left\{ \frac{V^{\rm eff}(\lceil \log_2 V\rceil + Q)}{B_k \log_2 \left(1+\frac{P_{k}^\mathrm{ul} |h_k|^2}{N_0^\mathrm{ser} B_k}\right)} \right\}.
\end{align}

For a nonempty \ac{DL} user set $\mathcal D$, the objective is to minimize the maximum user delay through \ac{DL} power allocation under the sum-power constraint:
\begin{subequations}\label{problem:p_dl}
    \begin{align}
        \min_{\mbf{p}_\mathcal{D}} \quad&  
        \max_{k\in\mathcal{D}} ~\left\{ \frac{ V^{\rm eff}(\lceil \log_2 V\rceil + Q)}{B_k \log_2 \left(1+\frac{p_{k} |h_k|^2}{N_{0, k}^\mathrm{dev} B_k}\right)} \right\}  \\ 
        \operatorname{ s.t. } \quad
        & p_k \geq 0, ~~\forall k \in \mathcal{D}, \\
        & \sum_{k \in \mathcal{D}} p_k \leq P^\mathrm{dl},
    \end{align}
\end{subequations}
where $\mbf{p}_{\mathcal{D}} \triangleq [p_k]_{k\in\mathcal{D}}$ collects the transmit powers for devices operating in the \ac{DL} mode.
The optimal \ac{DL} power allocation admits a closed-form solution parameterized by the resulting \ac{DL} delay $\tau^{\mathrm{dl}}$.
The result is summarized in the following lemma.
\begin{lemm}[Optimal \ac{DL} Power Allocation] \label{lemm:opt_dl_power_allocation}
    For any given nonempty \ac{DL} user set $\mathcal{D} \neq \emptyset$, problem \eqref{problem:p_dl} admits an optimal solution of the form
    \begin{align}
        p_k^\star = \frac{N_{0, k}^\mathrm{dev} B_k}{|h_k|^2} \left(2^\frac{V^{\rm eff}\left(\lceil \log_2 V\rceil + Q\right)}{B_k \tau^{\mathrm{dl}}} - 1\right), ~~\forall k \in \mathcal{D},
    \end{align}
    where $\tau^\mathrm{dl} > 0$ is the unique solution to
    \begin{align}
        \sum_{k \in \mathcal{D}} \frac{N_{0, k}^\mathrm{dev} B_k}{|h_k|^2} \left(2^\frac{V^{\rm eff}\left(\lceil \log_2 V\rceil + Q\right)}{B_k \tau^{\mathrm{dl}}} - 1\right) = P^\mathrm{dl}.
    \end{align}
    The resulting minimum \ac{DL} delay $\tau^{\mathrm{dl}}$ can be efficiently computed via a bisection search.
\end{lemm}
\begin{IEEEproof}
    The proof follows the max-delay equalization argument for power allocation under a sum-power constraint; see, e.g., classical results in convex optimization~\cite{boyd2004convex} and wireless max-min fairness power control~\cite{tse2005fundamentals}.
    The detailed derivations are omitted for brevity.
\end{IEEEproof}

\subsection{Optimal Mode Selection and Power Allocation}
We now turn to the optimization of the mode-selection vector $\mbf{m}$.
Although the original problem is binary and combinatorial, the optimal solution has a prefix structure after sorting the devices according to their minimum \ac{UL} delays.
Importantly, the mode-selection search is performed using the optimized delay values obtained from Lemma~\ref{lemm:opt_dl_power_allocation}; hence the power allocation is embedded in each mode-selection evaluation.

For each device $k$, define its minimum \ac{UL} delay as
\begin{align}
    \tau_k^\mathrm{ul} \triangleq \frac{V^{\rm eff}(\lceil \log_2 V\rceil + Q)}{B_k \log_2 \left(1+\frac{P_{k}^\mathrm{ul} |h_k|^2}{N_0^\mathrm{ser} B_k}\right)}.
\end{align}
Without loss of generality, upon observing the verification outcomes, we re-index the devices in $\mathcal R$ as $\{1,\ldots,|\mathcal R|\}$.
This allows us to focus exclusively on the subset of devices that require correction in the subsequent analysis.
We further sort them in non-decreasing order of their minimum \ac{UL} delays:
\begin{align}\label{eqn:ul_sorting}
    \tau_1^\mathrm{ul} \leq \tau_2^\mathrm{ul} \leq \dots \leq \tau_{|\mathcal{R}|}^\mathrm{ul}.
\end{align}
For any $n \in \{0,1,\dots, |\mathcal{R}|\}$, define
\begin{align}
    \mathcal{U} (n) & \triangleq \{1,2,\dots, n\},\\
    \mathcal{D} (n) & \triangleq \{n+1, n+2, \dots, |\mathcal{R}|\}.
\end{align}
That is, the first $n$ devices in the above ordering are assigned to the \ac{UL} mode, and the remaining $|\mathcal{R}|-n$ devices are assigned to the \ac{DL} mode.
For a given $n$, the minimum feasible \ac{UL} delay is
\begin{align}
    \tau^\mathrm{ul}(n) \triangleq \max_{k \in \mathcal{U}(n)} \tau_k^\mathrm{ul} = \tau_n^\mathrm{ul}, ~~ n \geq 1,
\end{align}
with the convention $\tau^\mathrm{ul}(0) = 0$.
Similarly, let $\tau^\mathrm{dl}(n)$ denote the optimized \ac{DL} delay obtained by applying Lemma~\ref{lemm:opt_dl_power_allocation} to the nonempty \ac{DL} user set $\mathcal D(n)$, with the convention $\tau^{\mathrm{dl}}(|\mathcal R|) = 0$.
The resulting optimized communication delay under this partition is therefore
\begin{align}
    \tau (n) = \max \{\tau^\mathrm{ul}(n), \tau^\mathrm{dl} (n)\}.
\end{align}

The following lemma shows that it suffices to restrict the search to this family of prefix-structured partitions to obtain the optimal mode-selection solution.

\begin{lemm}[Prefix Structure of Optimal Mode Selection]\label{lemm:prefix_mode_selection}
    There exists an optimal mode selection such that the \ac{UL} device set is of the prefix form $\mathcal{U}(n) = \{1,2, \dots, n\}$
    for some $n \in \{0, 1, \dots, |\mathcal{R}|\}$.
    The corresponding \ac{DL} device set is $\mathcal{D}(n)=\{n+1,n+2, \dots,|\mathcal{R}|\}$.
\end{lemm}

\begin{IEEEproof}
	Please refer to Appendix~\ref{proof:prefix_mode_selection}.
\end{IEEEproof}

\begin{figure}[t]
	\centering
	\includegraphics[width=.8\columnwidth]{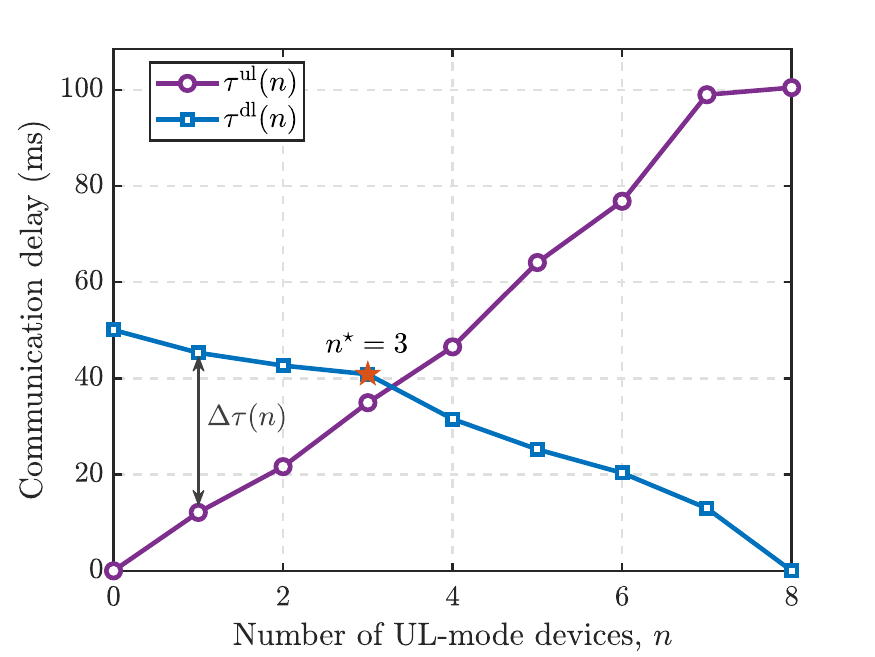}
     \vspace{-.3em}
	\caption{Illustration of the monotonic structure of \ac{UL} and \ac{DL} delays for a \ac{Multi-SPIN} round with $\left|\mathcal{R}\right| = 8$ devices participating in distribution transmission.
	The optimal \ac{UL}/\ac{DL} split point is determined by the intersection of the two delay curves, where the candidate solutions $n=3$ and $n=4$ are compared, yielding the optimal solution $n^\star=3$.}
	\label{fig:ul_dl_delays}
	\vspace{-.2em}
\end{figure} 
By Lemma~\ref{lemm:prefix_mode_selection}, the original combinatorial mode-selection problem can be reduced to a one-dimensional search over the \ac{UL}/\ac{DL} split point $n$.
This scalar search can be further reduced to a bisection search by exploiting the opposite monotonicity of the \ac{DL} and \ac{UL} delays.
To characterize this structure, define the \ac{DL}--\ac{UL} delay difference as
\begin{align}
    \Delta \tau(n) \triangleq \tau^{\mathrm{dl}}(n) - \tau^{\mathrm{ul}}(n).
\end{align}
As $n$ increases, more devices are assigned to the \ac{UL} mode.
Hence, $\tau^{\rm ul}(n)$ is non-decreasing in $n$, while $\tau^{\rm dl}(n)$ is non-increasing in $n$, because the \ac{DL} set becomes smaller and the same total \ac{DL} power budget is shared by fewer devices. Consequently, $\Delta \tau(n)$ is non-increasing in $n$.

\begin{prop}[Optimal Mode Selection via Bisection Search]
\label{prop:bisection_mode_selection}
    Let $\bar{n} \triangleq \min \, \{n:\Delta\tau(n)\leq 0\}$ be the first integer at which the delay difference becomes non-positive.
    Then the optimal \ac{UL}/\ac{DL} split point is given by
    \begin{align}\label{eqn:optimal_m}
    n^\star = 
    \arg\min_{n\in\{\bar n-1,\bar n\}\cap\{0,1,\dots,|\mathcal R|\}}
    ~\tau(n).
    \end{align}
    Moreover, $\bar n$ can be efficiently obtained via a bisection search over
$n\in\{0,1,\ldots,|\mathcal R|\}$.
\end{prop}

\begin{IEEEproof}
	Please refer to Appendix~\ref{proof:bisection_mode_selection}.
\end{IEEEproof}

Fig.~\ref{fig:ul_dl_delays} shows the \ac{UL} and \ac{DL} delays as functions of the number of \ac{UL}-mode devices.
The optimal mode split is determined by comparing the candidate values around the intersection point, as formalized in Proposition~\ref{prop:bisection_mode_selection}.

\begin{remk}[Coupling Between Optimal Mode Selection and Power Allocation]
    The bisection search in Proposition~\ref{prop:bisection_mode_selection} is performed over optimized delay values, rather than over fixed-power delay values.
    Specifically, for each queried value of $n$, the \ac{UL} delay $\tau^{\rm ul}(n)$ is obtained from maximum-power \ac{UL} transmission, while the \ac{DL} delay $\tau^{\rm dl}(n)$ is obtained by solving the \ac{DL} power-allocation problem in Lemma~\ref{lemm:opt_dl_power_allocation} over the current \ac{DL} set $\mathcal D(n)$.
    Thus, each evaluation of $\Delta\tau(n)$ implicitly invokes the optimal power-allocation solution.
    After the optimal $n^\star$ is determined, the corresponding optimal mode selection is given by the prefix structure $\mathcal U(n^\star)$ and $\mathcal D(n^\star)$, and the associated \ac{DL} powers are exactly those returned by Lemma~\ref{lemm:opt_dl_power_allocation}.
    Therefore, the proposed method jointly obtains the optimal mode selection and power allocation, rather than first fixing the modes and then optimizing the powers.
\end{remk}

\section{Proactive Draft-Length Control}\label{sec:proactive_draft_length}
In this section, we develop a greedy-search-based algorithm for draft-length control.
The main idea is to start from an initial draft-length vector and progressively increase the draft lengths until no candidate update can further improve the expected sum token goodput.
However, due to the stochastic nature of the objective and its dependence on the minimum expected communication latency, the resulting optimization problem is highly non-convex and cannot be solved in closed form.
Moreover, the evaluation of each candidate draft-length configuration requires solving the underlying reactive communication problem, which itself does not admit a closed-form solution.

\subsection{Algorithm Overview}\label{subsec:draft_length_alg_overview}
To address the above challenges, we design a greedy optimization procedure for draft-length control. 
The algorithm first performs synchronized all-ones updates to rapidly exploit global improvements in draft lengths across all devices.
After convergence of the first phase, a coordinate-wise greedy refinement is applied to capture residual gains by selectively increasing individual draft lengths.
The two phases are detailed as follows.
\subsubsection{Synchronized All-Ones Updates}
The algorithm first performs synchronized all-ones updates.
Let $\mathbf{1} \in \mathbb{R}^K$ denote the all-ones vector of length $K$.
Starting from the current draft-length vector $\boldsymbol{\gamma}$, the algorithm evaluates the candidate update
\begin{equation}
	\boldsymbol{\gamma}'
	=
	\boldsymbol{\gamma}
	+
	\mathbf{1}.
\end{equation}
Let $\chi^\star$ be the optimal reactive policy derived in Section~\ref{sec:reactive_mode_power}.
If $G(\boldsymbol{\gamma}', \chi^\star)
>
G(\boldsymbol{\gamma}, \chi^\star)$,
the update $\boldsymbol{\gamma}\leftarrow\boldsymbol{\gamma}'$ is accepted.
This procedure is repeated until no synchronized all-ones update can further improve the expected goodput.

\subsubsection{Coordinate Greedy Refinement}
After the synchronized phase terminates, the algorithm proceeds with coordinate greedy refinement.
Let $\mathbf e_k \in \mathbb{R}^K$ denote the $k$-th unit vector.
At each iteration, the algorithm evaluates the candidate updates
\begin{equation}
	\boldsymbol{\gamma}'_k
	=
	\boldsymbol{\gamma}
	+
	\mathbf e_k,
	\quad k\in\mathcal K.
\end{equation}
The device yielding the largest expected goodput is selected as
\begin{equation}
	k^\star
	=
	\arg\max_{k\in\mathcal K}~~
	G(\boldsymbol{\gamma}'_k, \chi^\star).
\end{equation}
If $G(\boldsymbol{\gamma}'_{k^\star}, \chi^\star)
>
G(\boldsymbol{\gamma},\chi^\star)$,
the update
$\boldsymbol{\gamma}\leftarrow\boldsymbol{\gamma}'_{k^\star}$ is accepted; otherwise, the coordinate greedy phase terminates.
The complete procedure is summarized in Algorithm~\ref{alg:greedy_draft_length}.
As detailed in Section~\ref{sec:expected_goodput_evaluation}, every evaluation of \(G(\boldsymbol{\gamma},\chi^\star)\) in Algorithm~\ref{alg:greedy_draft_length} is implemented using the \ac{SAA} estimate.

\begin{algorithm}[t]
	\caption{Greedy Search for Draft-Length Control}
	\label{alg:greedy_draft_length}
	\begin{algorithmic}[1]
		\STATE {\bfseries Input:} Optimal reactive policy $\chi^\star$
		\STATE {\bfseries Initialize:} $\bsm{\gamma} = \mbf{0}$
		\STATE \% \textit{Synchronized all-ones updates}
		\WHILE{$G(\bsm{\gamma} + \mathbf{1}, \chi^\star) > G(\bsm{\gamma}, \chi^\star)$}
		\STATE $\bsm{\gamma} \leftarrow \bsm{\gamma} + \mathbf{1}$
		\ENDWHILE
		
		\STATE \% \textit{Coordinate greedy refinement}
		\WHILE{$\exists k \in \mathcal{K}$ such that 
			$G(\boldsymbol{\gamma} + \mathbf{e}_k, \chi^\star) \!> G(\boldsymbol{\gamma}, \chi^\star)$}
		\STATE $k^\star = \arg\max_{k \in \mathcal{K}} G(\boldsymbol{\gamma} + \mathbf{e}_k, \chi^\star)$
		\STATE $\boldsymbol{\gamma} \leftarrow \boldsymbol{\gamma} + \mathbf{e}_{k^\star}$
		\ENDWHILE
		 \STATE {\bfseries Output:} Optimized draft lengths $\bsm{\gamma}$
	\end{algorithmic}
\end{algorithm}

\subsection{Expected Goodput Evaluation}\label{sec:expected_goodput_evaluation}
The remaining challenge is the evaluation of the expected goodput, which requires computing the expected minimum communication latency. 
Formally, we have
\begin{align}
	\mathbb{E}_{\mbf{o}|\bsm{\gamma}} \! \left[T^\mathrm{comm}_{\min}\left(\mbf{o}\right)\right] = \sum_{\mbf{o}\in\{0,1\}^K}\mathrm{Pr} (\mbf{o}\mid\bsm{\gamma}) T^\mathrm{comm}_{\min}\left(\mbf{o}\right),
\end{align}
where $\Pr(\mathbf{o}\mid\boldsymbol{\gamma}) = \prod_{k\in\mathcal K} \left(\alpha_k^{\gamma_k}\right)^{1-o_k} \left(1-\alpha_k^{\gamma_k}\right)^{o_k}$.
The above expression involves an exponential number of outcomes, since $\mbf{o} \in \{0,1\}^K$ yields $2^K$ possible realizations.
Moreover, the term $T^\mathrm{comm}_{\min}(\mbf{o})$ does not admit a closed-form expression due to its dependence on the underlying mode-selection and power-allocation decisions.

To overcome these difficulties, we adopt a \ac{SAA} approach.
Specifically, we draw $N_{\mathrm{SAA}}$ independent samples $\{\mbf{o}^{(\ell)}\}_{\ell=1}^{N_{\mathrm{SAA}}}$ according to $\mathrm{Pr}(\mbf{o}\mid\bsm{\gamma})$, and approximate the expectation as
\begin{align}\label{eqn:sample_avg}
	\mathbb{E}_{\mbf{o}|\bsm{\gamma}} \! \left[T^\mathrm{comm}_{\min}\left(\mbf{o}\right)\right]
	\approx
	\frac{1}{N_{\mathrm{SAA}}}
	\sum_{\ell=1}^{N_{\mathrm{SAA}}}
	T_{\min}^{\mathrm{comm}}
	\big(\mbf{o}^{(\ell)}\big).
\end{align}
By the law of large numbers, the sample average in~\eqref{eqn:sample_avg} converges to the expected minimum communication latency as $N_{\mathrm{SAA}}$ increases.
The \ac{SAA} approach therefore provides a consistent estimate while reducing the computational cost from enumerating all $2^K$ possible verification outcomes to evaluating only $N_{\mathrm{SAA}}$ sampled outcomes.
The resulting \ac{SAA} estimate is used to evaluate the goodput objective throughout Algorithm~\ref{alg:greedy_draft_length}.

\section{Experimental Results}\label{sec:experiments}
\subsection{Experimental Settings}
\subsubsection{Models and Tasks}
We consider two draft--target model pairs: Qwen2.5-0.5B with Qwen2.5-14B~\cite{qwen2024qwen25_14b_config}, and DeepSeek-R1-Distill-Qwen-1.5B with DeepSeek-R1-Distill-Qwen-32B~\cite{guo2025deepseek}.
For both model pairs, the draft and target models share the same tokenizer, with $V=151{,}665$ valid token IDs.\footnote{Although their model configurations report output dimensions of either 151{,}936 or 152{,}064, the additional dimensions correspond to padding introduced for distributed training rather than distinct tokenizer entries~\cite{qwen2024qwen25_14b_config, guo2025deepseek}.}
Each device is independently assigned one of four representative generation tasks uniformly drawn from HumanEval~\cite{dataset_human_eval}, GSM8K~\cite{dataset_gsm8k}, MT-Bench~\cite{dataset_mt_bench}, and IFEval~\cite{dataset_ifeval}, covering code generation, mathematical reasoning, multi-turn conversation, and instruction following, respectively.
For each model pair and task, we estimate the empirical token acceptance rate by running \ac{SPIN} on a set of sampled prompts and averaging the token-wise acceptance probabilities over all decoding rounds and prompts.
The resulting estimates are summarized in Table~\ref{table:acc_rate}.

\begin{table}[t]
	\caption{Empirical Acceptance Rate for Different Datasets and Model Pairs}
	\vspace{-.4em}
	\label{table:acc_rate}
	\centering
	\begin{tabular}{ccc}
		\toprule
		& Qwen2.5 pair  & DeepSeek-R1 pair \\
		\midrule
		HumanEval & 0.9453 & 0.8691   \\
		GSM8K & 0.8580 & 0.8618   \\
		MT-Bench & 0.8509 & 0.8221   \\
		IFEval & 0.8816 & 0.8323  \\
		\bottomrule
	\end{tabular}
	\vspace{-.4em}
\end{table}

\subsubsection{Computation Settings}
The edge server is equipped with an NVIDIA A100 GPU and executes the target model for batched draft verification. 
For each target model, we measure the verification latency at different batch sizes and fit the model in~\eqref{eqn:verification_latency}.
The fitted parameters are $T_0^\mathrm{ver}=42.8$ ms and $T_1^\mathrm{ver}=7.6$ ms for Qwen2.5-14B, and $T_0^\mathrm{ver}=69.3$ ms and $T_1^\mathrm{ver}=9.1$ ms for DeepSeek-R1-Distill-Qwen-32B.
The reference per-token drafting latencies of Qwen2.5-0.5B and DeepSeek-R1-Distill-Qwen-1.5B, measured on the same GPU, are 15.2 ms and 19.6 ms, respectively.
To capture heterogeneity in device-side computing capabilities, we independently scale the drafting latency of each device by a factor uniformly drawn from $[1,1.15]$, resulting in device-specific per-token drafting latencies $T_k^\mathrm{dft}$.
For the \ac{SAA}-based draft-length optimization, we use $N_{\mathrm{SAA}}=1000$ verification-outcome samples for each goodput evaluation.

\subsubsection{Communication Settings}
We consider an \ac{OFDMA} system with a total of $K$ devices sharing a bandwidth of $B=10$ MHz, where $K=20$ unless otherwise specified.
In each \ac{Multi-SPIN} round, the available bandwidth is equally allocated among the devices with rejected drafts.
We set the effective vocabulary size to $V^\mathrm{eff}=1024$ and quantize each probability value using $Q=16$ bits.
The wireless channels experience independent Rayleigh small-scale fading.
The large-scale channel gain is modeled using a reference gain of $-30$ dB at a distance of 1 m and a path-loss exponent of 3.5.
The device-to-server distances are independently drawn from the interval $[100,500]$ m.
All devices are assigned the same \ac{UL} power budget, denoted by $P^{\rm ul} \triangleq P_k^{\rm ul}$ for all $k \in \mathcal{K}$, whose value is specified in the corresponding figure.
The noise power spectral densities at the server and devices are set to $N_0^\mathrm{ser}=-169$ dBm/Hz and $N_0^\mathrm{dev}=-167$ dBm/Hz, respectively.

\subsubsection{Benchmarks}
We compare the proposed framework with several benchmark schemes to separately evaluate the performance gains from mode selection and draft-length control.
To assess the effectiveness of adaptive mode selection, we consider the following benchmarks:
\begin{itemize}
    \item \textbf{All \ac{UL} mode:}
    Devices with rejected drafts all operate in the \ac{UL} mode for draft-distribution uploading.
    Each device transmits at its maximum \ac{UL} power.
    \item \textbf{All \ac{DL} mode:}
    Devices with rejected drafts all operate in the \ac{DL} mode for target-distribution downloading.
    The \ac{DL} transmit power is optimally allocated according to Lemma~\ref{lemm:opt_dl_power_allocation}.
    \item \textbf{Random mode selection:}
    Devices with rejected drafts independently select the \ac{UL} or \ac{DL} mode with equal probability.
    Given the resulting mode assignment, the \ac{UL}-mode devices transmit at their maximum powers, while the \ac{DL} power is optimally allocated among the \ac{DL}-mode devices.
\end{itemize}
For each mode-selection scheme, the draft lengths are separately optimized under the corresponding communication policy.
This ensures that the comparison isolates the performance gain of adaptive mode selection rather than that of draft-length control.

For draft-length control, we consider the following benchmarks:
\begin{itemize}
    \item \textbf{Optimal uniform draft length:}
    All devices use the same draft length, i.e., $\gamma_k=\gamma$ for all $k\in\mathcal K$.
    The uniform draft length $\gamma$ is selected by exhaustive search over $\{0, 1, \dots, \gamma_{\max}\}$ to maximize the expected sum token goodput, where we set $\gamma_{\max}= 20$.
    
    \item \textbf{Computation-only draft-length control:}
    The draft lengths are selected by considering only the computation part of the round latency, while ignoring the communication delay.

\end{itemize}
For both draft-length schemes, the reactive communication stage employs the optimal mode-selection and power-allocation policy developed in Section~\ref{sec:reactive_mode_power}.
Therefore, their performance differences arise solely from the draft-length control strategy.

\begin{figure}[t]
	\centering
	\begin{minipage}{0.4939\linewidth}
		\centering
		\includegraphics[width=1\linewidth]{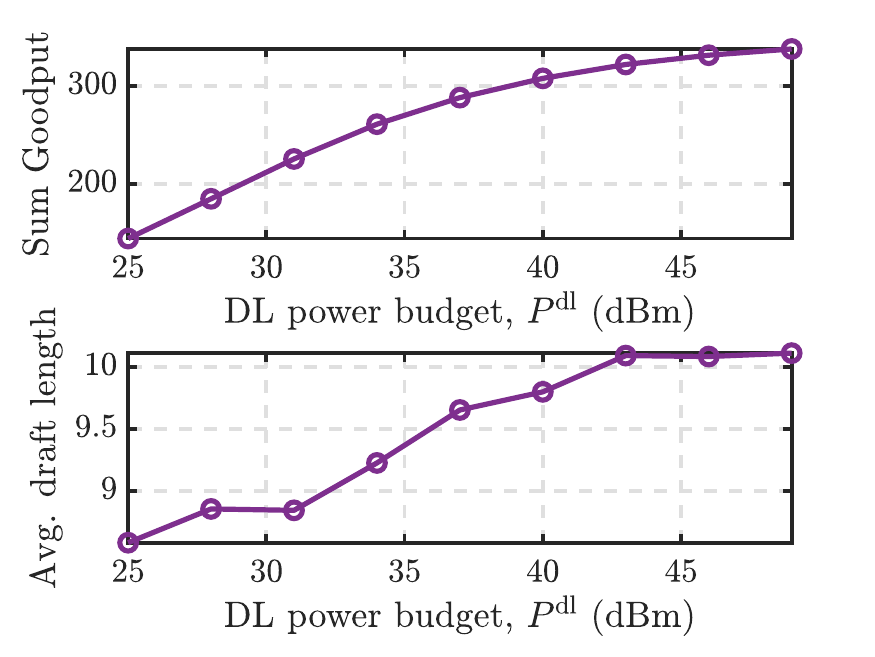}\\
		\footnotesize{(a) Draft-length behavior}
	\end{minipage}
	\begin{minipage}{0.4939\linewidth}
		\centering		
		\includegraphics[width=1\linewidth]{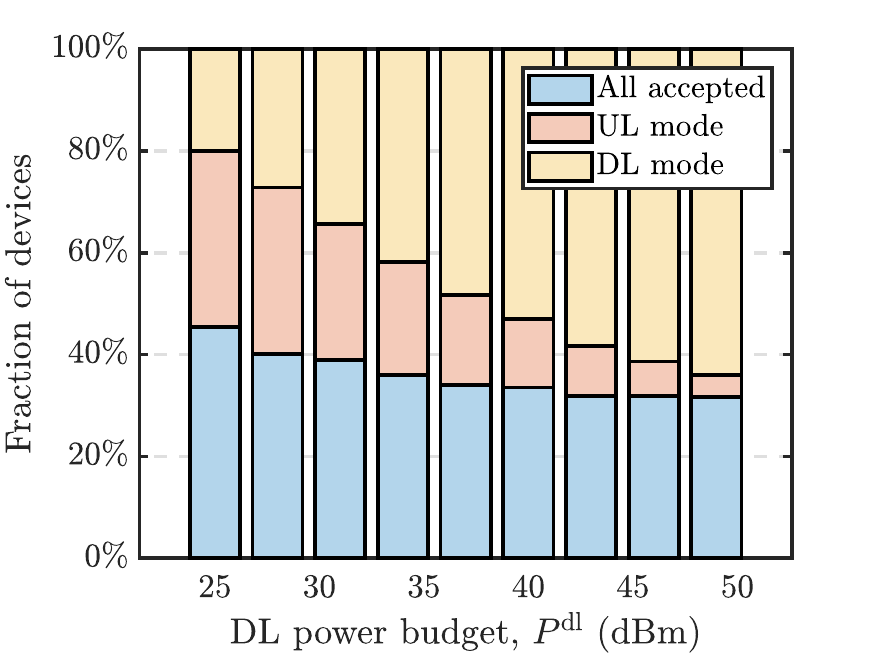}\\
		\footnotesize{(b) Mode-selection behavior}
	\end{minipage}
	\caption{\label{fig:solution_qwen}
		Adaptive draft-length control and mode selection versus the \ac{DL} power budget for Qwen2.5 model pair, with the per-device \ac{UL} power budget fixed at $P^\mathrm{ul} = 10$ dBm.}
\end{figure}

\begin{figure}[t]
	\centering
	\begin{minipage}{0.4939\linewidth}
		\centering
		\includegraphics[width=1\linewidth]{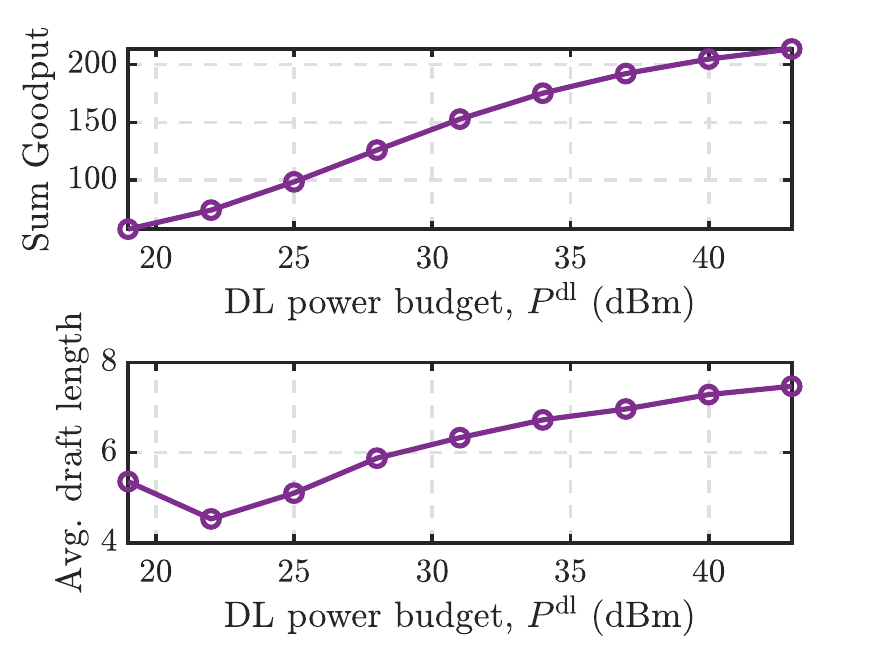}\\
		\footnotesize{(a) Draft-length behavior}
	\end{minipage}
	\begin{minipage}{0.4939\linewidth}
		\centering		
		\includegraphics[width=1\linewidth]{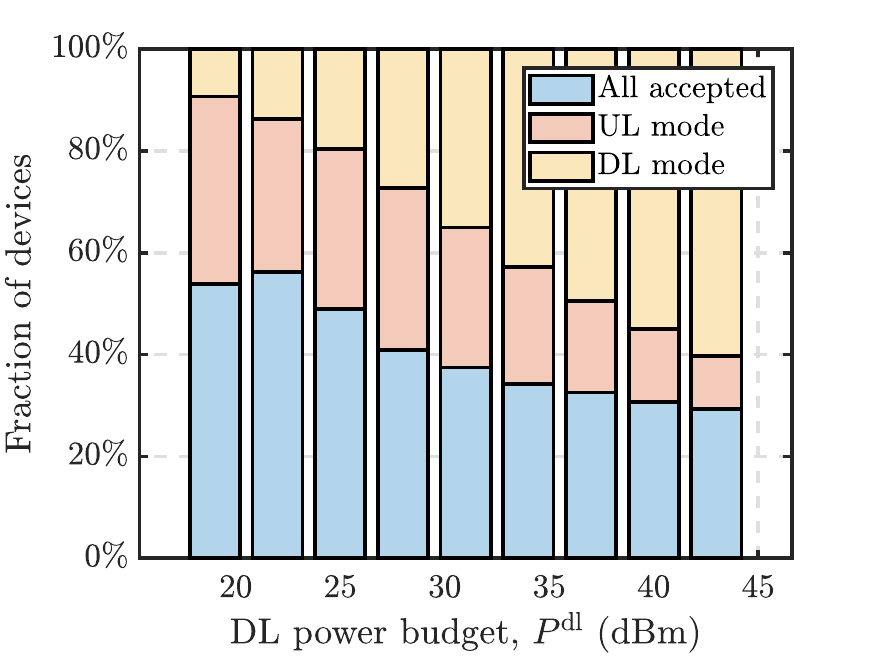}\\
		\footnotesize{(b) Mode-selection behavior}
	\end{minipage}
	\caption{\label{fig:solution_structure_ds}
		Adaptive draft-length control and mode selection versus the \ac{DL} power budget for DeepSeek-R1 model pair, with the per-device \ac{UL} power budget fixed at $P^\mathrm{ul} = 10$ dBm.}
\end{figure}

\subsection{Adaptive Draft-Length and Mode-Selection Behaviors}
Figs.~\ref{fig:solution_qwen} and~\ref{fig:solution_structure_ds} illustrate how the proactive draft-length control and the reactive mode selection adapt to the \ac{DL} power budget $P^{\rm dl}$.
For both model pairs, the sum token goodput increases steadily with $P^{\rm dl}$.
As $P^{\rm dl}$ becomes sufficiently large, communication delay becomes less dominant.
The performance bottleneck then shifts toward local drafting and server-side verification, causing the goodput gain to gradually saturate.
The same bottleneck shift also explains the draft-length behavior.
The average optimized draft length generally increases with the \ac{DL} power budget.
A longer draft provides greater speculative gain by exposing more tokens to parallel verification, but also increases the probability of rejection and hence the likelihood of costly distribution transmission.
As $P^{\rm dl}$ increases, the communication penalty associated with rejection is reduced.
The draft-length controller can therefore tolerate more frequent rejections and select longer drafts to better exploit the parallel-verification gain.

\begin{figure}[t]
	\centering
	\begin{minipage}{0.4939\linewidth}
		\centering
		\includegraphics[width=1\linewidth]{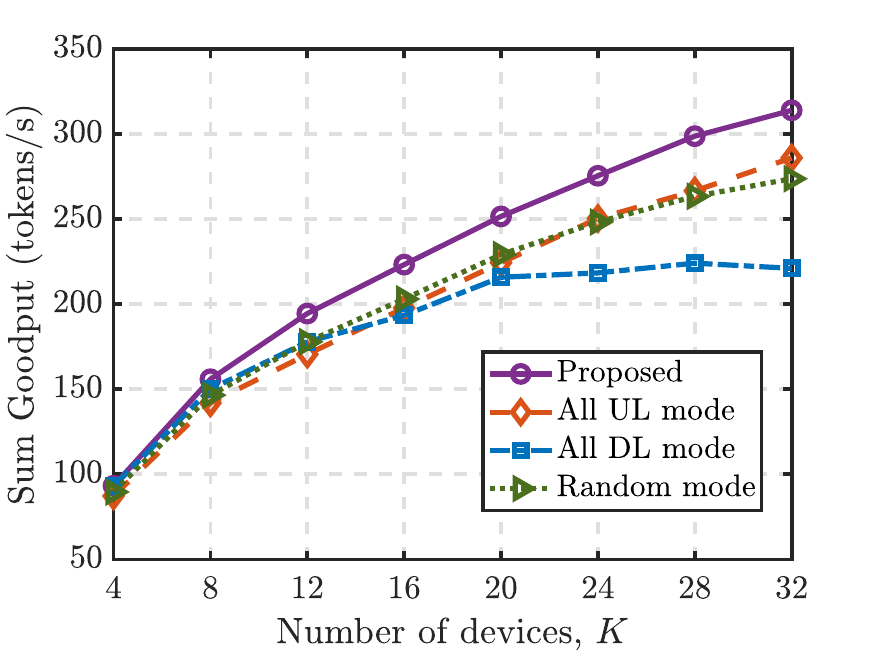}\\
		\footnotesize{(a) $P^{\rm dl}/P^{\rm ul}=4$}
	\end{minipage}
	\begin{minipage}{0.4939\linewidth}
		\centering		
		\includegraphics[width=1\linewidth]{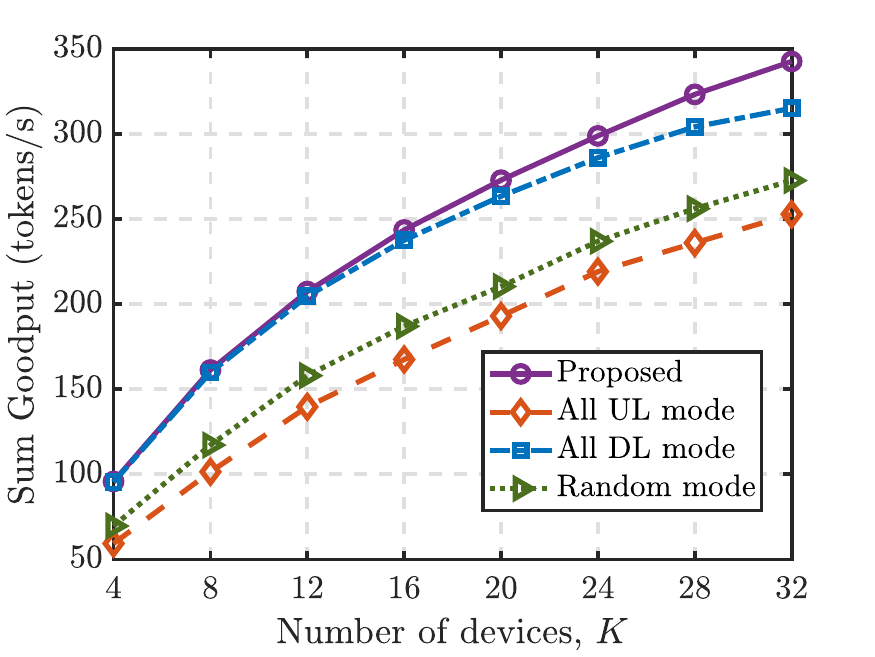}\\
		\footnotesize{(b) $P^{\rm dl}/P^{\rm ul}=16$}
	\end{minipage}
	\caption{\label{fig:number_devices_qwen}
		Sum goodput versus the number of devices for Qwen2.5 model pair under different \ac{DL}-to-\ac{UL} power budget ratios.}
\end{figure}

\begin{figure}[t]
	\centering
	\begin{minipage}{0.4939\linewidth}
		\centering
		\includegraphics[width=1\linewidth]{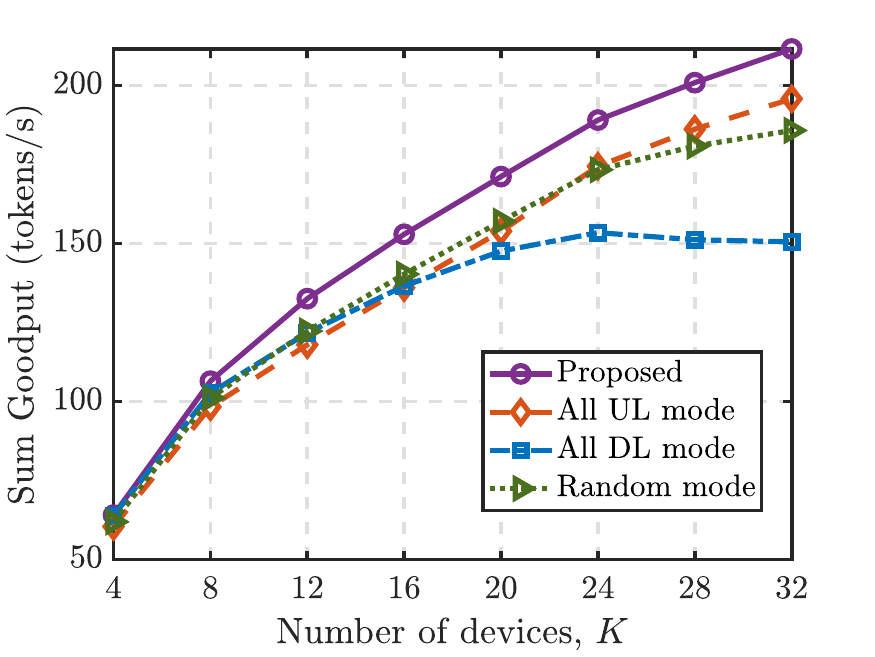}\\
		\footnotesize{(a) $P^{\rm dl}/P^{\rm ul}=4$}
	\end{minipage}
	\begin{minipage}{0.4939\linewidth}
		\centering		
		\includegraphics[width=1\linewidth]{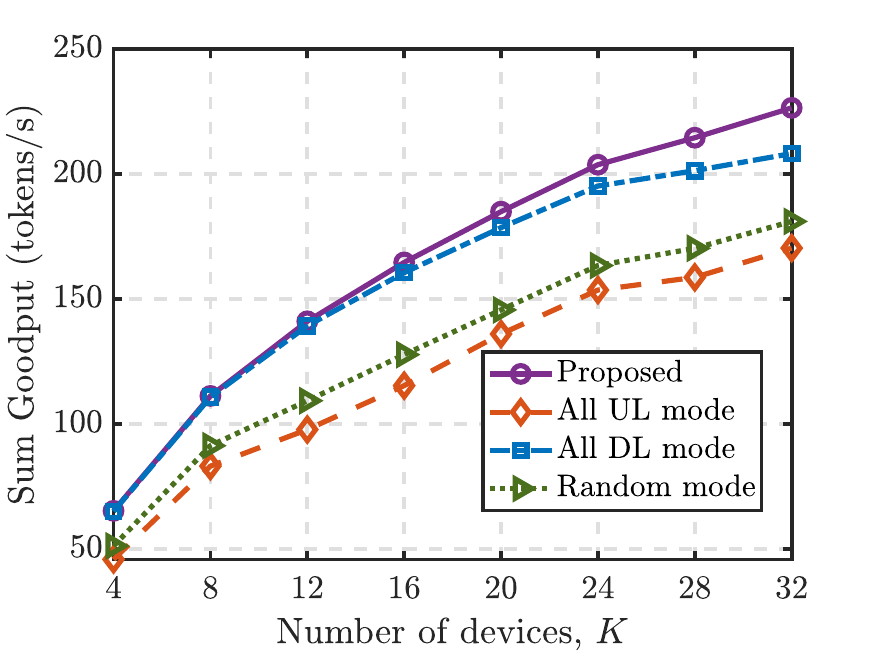}\\
		\footnotesize{(b) $P^{\rm dl}/P^{\rm ul}=16$}
	\end{minipage}
	\caption{\label{fig:number_devices_ds}
		Sum goodput versus the number of devices for DeepSeek-R1 model pair under different \ac{DL}-to-\ac{UL} power budget ratios.}
\end{figure}

Figs.~\ref{fig:solution_qwen} and~\ref{fig:solution_structure_ds} further show the fraction of devices in the three possible states: all drafted tokens accepted, correction through the \ac{UL} mode, and correction through the \ac{DL} mode.
When $P^{\rm dl}$ is small, the shared \ac{DL} resource is restrictive, and most devices requiring correction are therefore assigned to the \ac{UL} mode.
As $P^{\rm dl}$ increases, the optimized \ac{DL} transmission delay decreases, and the proposed scheme progressively shifts rejected devices from the \ac{UL} mode to the \ac{DL} mode.
Nevertheless, a fraction of devices continues to use the \ac{UL} even at relatively large $P^{\rm dl}$, since assigning every rejected device to the \ac{DL} would introduce unnecessary competition for the shared \ac{DL} power budget.

\subsection{Performance Gain of Optimal Mode Selection}

\begin{figure}[t]
	\centering
	\begin{minipage}{0.4939\linewidth}
		\centering
		\includegraphics[width=1\linewidth]{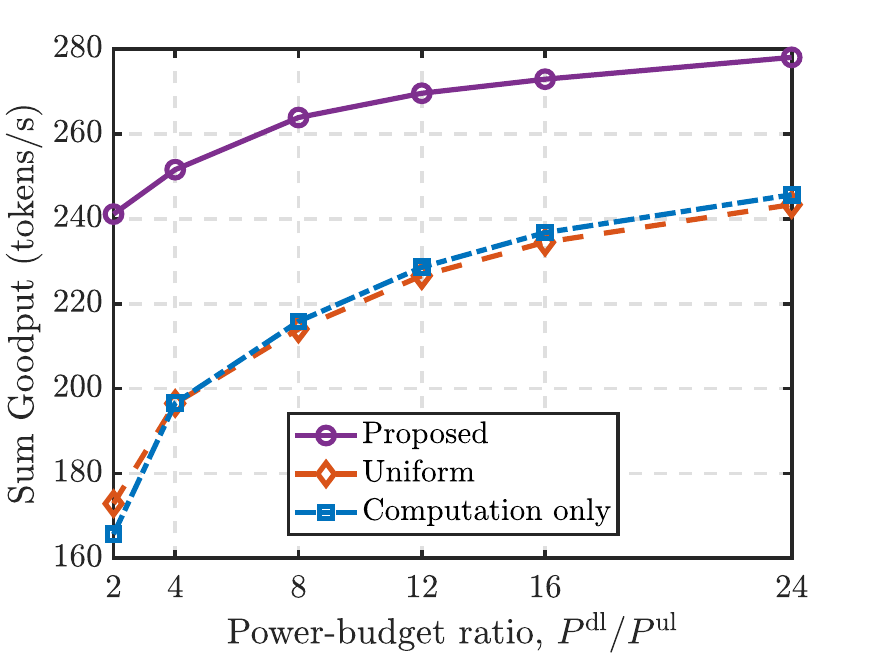}\\
		\footnotesize{(a) Qwen2.5 pair}
	\end{minipage}
	\begin{minipage}{0.4939\linewidth}
		\centering		
		\includegraphics[width=1\linewidth]{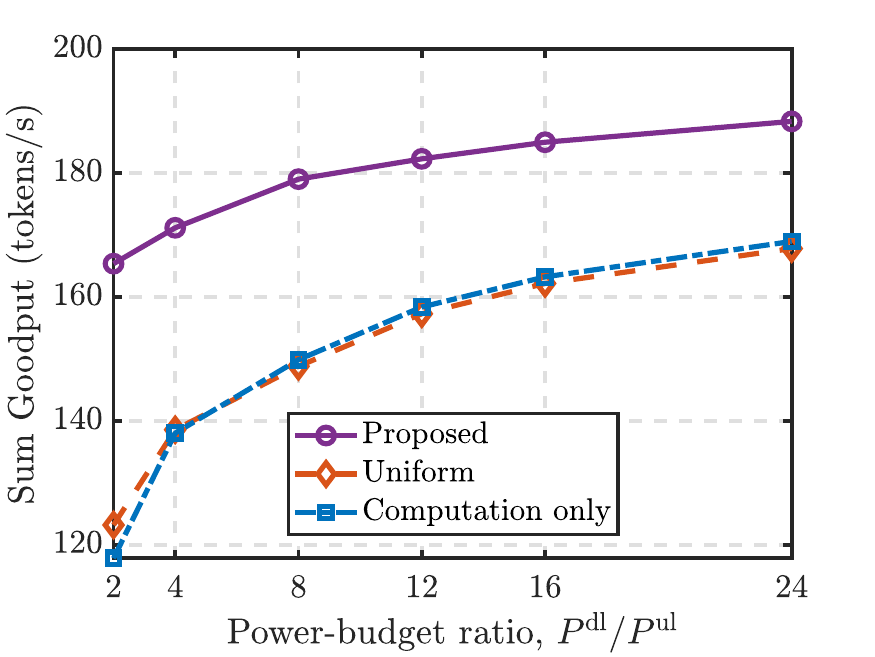}\\
		\footnotesize{(b) DeepSeek-R1 pair}
	\end{minipage}
	\caption{\label{fig:draft_lengths}
		Sum goodput versus the \ac{DL}-to-\ac{UL} power-budget ratio for Qwen2.5 and DeepSeek-R1 model pairs.}
\end{figure}

Figs.~\ref{fig:number_devices_qwen} and~\ref{fig:number_devices_ds} compare the goodput of different mode-selection schemes as the number of devices increases.
For each $K$, the total communication power budget is fixed as $P^{\rm tot}(K) = \bar{P}^{\rm dl} + K \bar{P}^{\rm ul}$, where $\bar{P}^{\rm dl} = 30$ dBm and $\bar{P}^{\rm ul} = 15$ dBm are the reference \ac{DL} and per-device \ac{UL} power budgets, respectively.
The actual budgets $P^{\rm dl}$ and $P^{\rm ul}$ are selected to satisfy both the prescribed \ac{DL}-to-\ac{UL} power-budget ratio $P^{\rm dl}/P^{\rm ul}$ and the total-budget constraint $P^{\rm dl}+K P^{\rm ul}=P^{\rm tot}(K)$.
The left and right subfigures correspond to $P^{\rm dl}/P^{\rm ul} = 4$ and $P^{\rm dl}/P^{\rm ul} = 16$, respectively.
In both cases, the proposed optimal mode-selection scheme consistently achieves the highest goodput by adaptively balancing the \ac{UL} and \ac{DL} modes.
When the \ac{DL} power is relatively limited, i.e., $P^{\rm dl}/P^{\rm ul}=4$, assigning all devices to the \ac{DL} mode leads to severe competition for the shared server power, while the all-\ac{UL} mode is constrained by the per-device transmit-power budget.
The proposed scheme avoids both extremes by assigning only a suitable subset of devices to the \ac{DL} mode.
When the \ac{DL} power is more abundant, i.e., $P^{\rm dl}/P^{\rm ul}=16$, the all-\ac{DL} scheme becomes more competitive, but remains inferior to the proposed scheme, particularly as the number of devices increases and competition for the shared \ac{DL} power intensifies.
These results demonstrate that neither a fixed \ac{UL} nor a fixed \ac{DL} policy is universally optimal, and that adaptive mode selection is essential for maintaining high goodput in Multi-\ac{SPIN}.

\subsection{Performance Gain of Proactive Draft-Length Control}

Fig.~\ref{fig:draft_lengths} evaluates the benefit of the proposed proactive draft-length control under different \ac{DL}-to-\ac{UL} power-budget ratios.
The proposed scheme consistently achieves higher goodput than both the optimal uniform draft length and the computation-only draft-length control.
Compared with the uniform-draft baseline, the gain highlights the benefit of adapting draft lengths to device-specific latency heterogeneity.
Compared with the computation-only baseline, it further demonstrates the importance of accounting for communication conditions in draft-length optimization.

\section{Concluding Remarks}\label{sec:conclusions}
In this paper, we introduced communication-mode selection for \ac{Multi-SPIN}.
The underlying philosophy of \ac{Multi-SPIN} is to move beyond server-centric inference by distributing draft generation across local \acp{SLM}, thereby harnessing otherwise underutilized device-side computation for cooperative token generation.
Communication-mode selection extends this philosophy from computation to communication.
Rather than relying exclusively on the shared \ac{DL}, the system can also exploit the distributed \ac{UL} capabilities of edge devices and select the more efficient transmission direction for each correction operation.
This differs fundamentally from conventional communication systems, in which the information source, destination, and transmission direction are typically predetermined.
In \ac{Multi-SPIN}, the direction of information flow can instead be adapted to system conditions, creating an additional degree of freedom for communication design.
Wireless communication thus evolves from a passive data-delivery link into an active component of edge \ac{LLM} inference.

This perspective opens several promising research directions.
One is to move beyond binary \ac{UL}/\ac{DL} mode selection toward multi-mode designs that jointly choose what information to exchange, in which direction, and at what representation granularity, such as full, truncated, quantized, or compressed token distributions.
Another opportunity lies in jointly designing communication modes with model collaboration itself, including draft-model assignment, draft length, and early bypass of either local drafting or server verification.
These developments could ultimately lead to a unified framework in which computation placement, information representation, and transmission direction are co-designed for scalable and resource-efficient edge intelligence.

\appendices
\section{Proof of Lemma~\ref{lemm:prefix_mode_selection}} \label{proof:prefix_mode_selection}
Consider any feasible mode selection whose \ac{UL} set $\mathcal U$ is not of prefix form.
Then there exist two devices $i<j$ such that
$i\notin \mathcal{U}$ but $j\in \mathcal{U}$.
Since the devices are sorted according to their minimum
\ac{UL} delays, we have $\tau_i^{\mathrm{ul}} \leq \tau_j^{\mathrm{ul}}$.

Now construct a new mode selection by moving device $i$ from \ac{DL} to \ac{UL}, while keeping all other devices unchanged.
The resulting UL set is $\mathcal U'=\mathcal U\cup\{i\}$, and the resulting DL set is $\mathcal D'=\mathcal D\setminus\{i\}$.
Under this change, the minimum feasible \ac{UL} delay does not increase, because device $i$ has no larger \ac{UL} delay than device $j$.
Meanwhile, the minimum feasible \ac{DL} delay cannot increase, because the \ac{DL} power-allocation problem is solved over the smaller user set $\mathcal D'$ under the same total \ac{DL} power budget.
Therefore, the resulting delay is no larger than that of the original mode selection.

Repeating this operation whenever there exists a pair $i<j$ such that $i\notin\mathcal U$ but $j\in\mathcal U$ eventually yields a prefix \ac{UL} set $\mathcal U(n)=\{1,2,\ldots,n\}$ for some $n\in\{0,1,\ldots,|\mathcal R|\}$, without increasing the communication delay.
Hence, there exists an optimal mode selection with the claimed prefix structure.

\section{Proof of Proposition~\ref{prop:bisection_mode_selection}} \label{proof:bisection_mode_selection}
Since $\tau^{\rm ul}(n)$ is non-decreasing in $n$ and $\tau^{\rm dl}(n)$ is non-increasing in $n$, the delay difference $\Delta \tau(n)$ is non-increasing in $n$.
Therefore, the first index $\bar n$ satisfying $\Delta\tau(\bar n)\leq 0$ can be found via a bisection search.

It remains to show that the global minimizer of $\tau(n)$ must lie around this crossing point.
For any $n<\bar n$, we have $\Delta\tau(n)>0$, which implies $\tau^{\rm dl}(n)>\tau^{\rm ul}(n)$.
Hence, the overall delay is dominated by the \ac{DL} part, i.e., $\tau(n)=\tau^{\rm dl}(n)$.
Since $\tau^{\rm dl}(n)$ is non-increasing in $n$, the best choice among all $n<\bar n$ is $n=\bar n-1$.
Similarly, for any $n\geq \bar n$, we have $\Delta\tau(n)\leq 0$, which implies $\tau^{\rm ul}(n)\geq \tau^{\rm dl}(n)$.
Hence, the overall delay is dominated by the \ac{UL} part, i.e., $\tau(n)=\tau^{\rm ul}(n)$.
Since $\tau^{\rm ul}(n)$ is non-decreasing in $n$, the best choice among all $n\geq \bar n$ is $n=\bar n$.
Therefore, the optimal solution must be one of the two candidates in~\eqref{eqn:optimal_m}.

\bibliographystyle{IEEEtran}
\bibliography{IEEEabrv,mybib}

\end{document}